\documentclass[11pt]{article}

\usepackage[a4paper,margin=1in]{geometry}
\usepackage[T1]{fontenc}
\usepackage[utf8]{inputenc}
\usepackage{lmodern}
\usepackage{amsmath,amssymb,amsthm,mathtools}
\usepackage{array}
\usepackage{enumitem}
\usepackage{float}
\usepackage{graphicx}
\usepackage{microtype}
\usepackage{placeins}
\usepackage{xspace}
\usepackage{xcolor}
\usepackage{hyperref}
\usepackage[nameinlink,capitalise]{cleveref}
\usepackage{abstract}

\hypersetup{
  colorlinks=true,
  linkcolor=blue!45!black,
  citecolor=blue!45!black,
  urlcolor=blue!45!black,
  pdftitle={Residual-Entropy Accounting for Routed Atom-Budgeted Learned Indexes},
  pdfauthor={Faruk Alpay and Levent Sarioglu}
}

\newcommand{\U}{\mathcal{U}}
\newcommand{\M}{\mathcal{M}}
\newcommand{\Dir}{\mathsf{Dir}}
\newcommand{\Comp}{\mathsf{Comp}}
\newcommand{\Gap}{\mathsf{Gap}}
\newcommand{\RGap}{\mathsf{RGap}}
\newcommand{\Lin}{\mathsf{Lin}}
\newcommand{\Rep}{\mathsf{Rep}}

\newcommand{\Def}{\mathsf{Def}}

\newcommand{\diam}{\mathsf{diam}}
\newcommand{\Beff}{B_{\mathrm{eff}}}
\newcommand{\rank}{\mathsf{rank}}
\newcommand{\pred}{\mathsf{pred}}
\newcommand{\E}{\mathbb{E}}
\newcommand{\Prb}{\mathbb{P}}

\newcommand{\logplus}{\log^{+}}
\newcommand{\KL}{D_{\mathrm{KL}}}
\DeclareMathOperator*{\argmin}{arg\,min}

\theoremstyle{plain}
\newtheorem{theorem}{Theorem}
\newtheorem{lemma}{Lemma}
\newtheorem{proposition}{Proposition}
\newtheorem{corollary}{Corollary}

\theoremstyle{definition}
\newtheorem{definition}{Definition}

\newtheorem{remark}{Remark}
\newtheorem{example}{Example}

\title{\bfseries Residual-Entropy Accounting for Routed Atom-Budgeted Learned Indexes}
\author{
Faruk Alpay\thanks{Corresponding author: \texttt{alpay@lightcap.ai}.}\\
\small Department of Computer Engineering, Bah\c{c}e\c{s}ehir University\\[-0.1em]
\small Istanbul, Turkey\\[-0.1em]
\small \texttt{faruk.alpay@bahcesehir.edu.tr}
\and
Levent Sar\i{}o\u{g}lu\\
\small Department of Computer Engineering, Bah\c{c}e\c{s}ehir University\\[-0.1em]
\small Istanbul, Turkey\\[-0.1em]
\small \texttt{levent.sarioglu@bahcesehir.edu.tr}
}
\date{}

\begin{document}
\maketitle
\begin{abstract}
We study exact predecessor and rank search within a routed, atom-budgeted, certified-repair learned-index architecture: an ordered directory routes a query to a contiguous interval, a counted local predictor returns a certified rank window, and exact repair resolves the residual uncertainty by comparisons.  This scope is part of the claim from the outset; the result is not a guarantee for arbitrary learned-index designs such as unconstrained RMI dispatch, hash routing, neural routing, or exact-payload indexes without additional accounting.  The main parameter is conditional residual answer entropy, namely the entropy of the exact answer after the leaf, predictor output, certificate, and charged pre-repair information have been observed.  We prove a two-sided accounting theorem showing that this functional gives the query-time scale, under the stated architecture and local predictor-atom budget.  Directory space, sorted-array storage, and transcript-indexed repair-program space are separate system costs, so the theorem is not a byte-level space lower bound and is not an implementation recipe for a compact repair-program table; in a materialized upper-bound construction the repair-program table can be as large as the sum of the attainable residual answer sets over transcript values.  The radius expression based on $\log(1+\Delta)$ is a rank-spread specialization, valid only when the predictor transcript leaves many residual ranks with non-negligible conditional probability.  We make the profile term non-oracular for counted piecewise-linear segments, derive a shadow-price allocation rule, compute exact finite-instance $\RGap_\M$ and $\Gap_\M$ values on real SOSD/Zenodo samples for an explicit affine atom model, and report end-to-end benchmark results against PGM-index, RadixSpline, and binary search.  The systems benchmarks expose overheads and bottlenecks; they are not a speed claim for the shadow prototype.
\end{abstract}
\vspace{0.75em}

\section{Introduction}

The classical theory of ordered dictionaries treats the input order as adversarial unless a distribution is explicitly built into the search tree.  Balanced search trees give $O(\log n)$ comparisons; predecessor structures improve the dependence on the universe in word-RAM regimes; and cell-probe lower bounds explain why broad classes of static and dynamic ordered search problems remain hard under space restrictions \cite{vanemdeboas1977veb,ajtai1988predecessors,beame2002predecessor,patrascu2006predecessor,patrascu2007randomization,navarro2020predecessor}.  Learned indexes alter the premise.  They exploit the fact that many sorted datasets have rank functions with large predictable regions, so that a small stored model can place a query close to its true rank before exact correction begins.

The systems literature has developed this idea rapidly.  Recursive model indexes, FITing-Tree, RadixSpline, ALEX, LIPP, and the PGM-index demonstrate that prediction can replace a substantial part of a comparison tree on regular data \cite{kraska2018case,galakatos2019fiting,kipf2020radixspline,ding2020alex,wu2021lipp,ferragina2020pgm}.  Benchmarks and stress tests also show that the benefit depends sharply on local error, corrective search, memory layout, updates, and robustness under distribution shift \cite{marcus2021benchmarking,maltry2022critical,wongkham2022ready,liu2025pgmpp,zhang2024disk,luo2025robust}.  This paper isolates the common comparison-routing skeleton behind segment-based systems: the router chooses a local model, the model returns a certified position window, and the last-mile repair is exact.

A satisfactory theory should identify that structure without reducing the problem to either a worst-case lower bound or an average model-error statistic.  Distribution-sensitive search trees already show that non-uniform query mass is worth entropy rather than $\log n$ \cite{knuth1971obst,mehlhorn1975obst,bent1985biased,sleator1985splay}.  Approximation-theoretic learned-index results show that model families have intrinsic limits on how cheaply they approximate a rank curve \cite{zeighami2023distribution,croquevielle2025constant,croquevielle2026lower}.  What is missing is a single parameter that couples these two phenomena: a key interval should be cheap only if it is easy to reach, easy to approximate, or rarely queried.

We study this coupling for \emph{routed piecewise learned indexes}.  A query first follows an ordered binary directory to a contiguous key interval, then evaluates a local predictor, and finally performs exact repair inside a certified rank window.  This is the theorem's architecture.  It covers the segment-based learned-index skeleton used by PGM-like, FITing-tree-like, and spline-like systems when their routing and local records are made explicit.  It does not directly cover unconstrained recursive model indexes with constant-time arithmetic dispatch, learned hash-routing layers, neural routers, or exact-payload structures unless their dispatch or payload information is charged in the transcript or in a separate memory term.  This limitation is central to the contribution: the paper identifies a residual-entropy accounting objective for a concrete architecture, not a theorem about arbitrary learned-index designs.

The novelty is not the isolated use of Shannon entropy, alphabetic coding, or conditional decision trees.  Those tools are classical.  The contribution is the architecture-specific accounting that makes them apply to learned indexes: a local atom budget for predictive records, explicit non-atom accounting for directories and repair programs, certified repair windows, and a residual objective that couples routing mass with the answer entropy left after prediction.  This accounting is what turns familiar information-theoretic ingredients into a learned-index instance parameter.

\begin{center}
\fbox{\begin{minipage}{0.91\linewidth}
\textbf{What the atom budget is not.}
The budget $B$ is not a total-space budget.  It counts only local predictive records.  Directory topology, directory keys, the sorted array, and transcript-indexed repair programs are outside $B$.  The upper bound in \cref{thm:residual-law} may use such a repair program as an information-theoretic coding device; if transcript value $y$ leaves $K_y$ exact answers attainable in the required query domain, a materialized correct table can require $\Theta(\sum_y K_y)$ repair-program nodes and has no bound in terms of $B$ alone.  A byte-space theorem would therefore need an additional directory/repair-program memory term.  The experiments report those bytes and routing costs separately instead of claiming a byte-level space guarantee.
\end{minipage}}
\end{center}

\subsection*{Technical obstruction}

A worst-case prediction radius $\Delta$ does not, by itself, imply an expected lower bound of $\log(1+\Delta)$.  A workload may put all its mass on a single residual answer inside the window.  Conversely, an empirical average error does not certify a lower-bound parameter, because another partition might spend the same number of atoms in different intervals.  The correct lower-bound object is the residual entropy of the exact answer after the predictor and its certificate are known.  The familiar $\log(1+\Delta)$ repair term emerges when the conditional workload is rank-spread across the certified window.

The rank-spread distinction appears in ordinary workloads.  It is plausible for uniform local lookups, timestamp buckets with jitter, or spatial traces that leave many nearby OSM gaps possible after prediction.  It fails for repeated hot-key lookups, endpoint probes, and cache-filtered workloads where the same transcript almost always maps to one residual rank.

For a partition $\Pi=\{I_1,\ldots,I_m\}$ with masses $p_j=\mu(I_j)$, predictors $h_j$, and certificates $\Delta_j$, the canonical lower-bound object is
\[
  H_\mu(\Pi)+\sum_{j=1}^{m}p_j\Rep_\mu(I_j,h_j,\Delta_j),
\]
where $\Rep_\mu$ is the residual entropy of the exact answer in leaf $I_j$.  Only after imposing rank-spread does this expression reduce to the more familiar radius surrogate
\[
  \sum_{j=1}^{m}p_j\left(\log\frac{1}{p_j}+\log(1+\Delta_j)\right).
\]
The feasible certificates are constrained by local approximation profiles $\Comp_{I,\M}(\Delta)$.  Thus the paper has two levels: a general residual-entropy theorem and a rank-spread specialization expressed through $\Gap_{\M}(S,\mu,B)$.

\subsection*{Contributions}

\begin{itemize}[leftmargin=1.5em]
  \item We formalize a budgeted atom model for routed piecewise learned indexes.  The formalism separates local prediction from directory bytes, sorted-array storage, stored repair programs, and exact repair comparisons.
  \item \textbf{Theorem 1} formalizes the residual-entropy functional $\RGap_\M(S,\mu,B)$ and proves two-sided query-time bounds for it under this architecture and accounting convention.
  \item \textbf{The rank-spread corollary} defines the radius surrogate $\Gap_\M(S,\mu,B)$ only as a specialization of $\RGap_\M$.  Without rank-spread the residual-entropy theorem is the correct statement and $\Gap_\M$ is not a valid lower bound.
  \item We make the local profile $\Comp_{I,\M}(\Delta)$ concrete for error-bounded piecewise-linear segment families, where it equals the minimum certified segment cover of the local rank curve.
  \item We develop a discrete shadow-price calculus for atom allocation.  For normalized power-law profiles, including integer atoms, atom floors, and boundary radii, it gives a water-filling rule and a hard-mass obstruction to constant expected query time.
  \item We provide a running example, synthetic allocation calculations, exact finite-instance $\RGap_\M/\Gap_\M$ measurements on real-data samples, and a reproducible systems benchmark artifact using large SOSD/Zenodo datasets with PGM-index, RadixSpline, binary search, and a diagnostic atom-allocation prototype.
\end{itemize}

\subsection*{Organization}

\Cref{sec:related} positions the paper relative to ordered search, succinctness, dynamic lower bounds, and learned-index systems.  \Cref{sec:model} defines atoms, local profiles, certificates, piecewise-linear profile computation, and rank-spread workloads.  \Cref{sec:calculus} develops the shadow-price calculus for the radius surrogate.  \Cref{sec:static} proves the residual-entropy accounting law and then derives the rank-spread $\Gap_\M$ specialization.  \Cref{sec:powerlaw} gives the discrete power-law specialization, examples, and design rules.  \Cref{sec:experiments} reports a real-data benchmark artifact and diagnostics.  \Cref{sec:dynamic} records dynamic boundary checks and limitations.  Appendices contain expanded proofs and rounding details.

\subsection*{Notation guide}

\begin{center}
\small
\begin{tabular}{>{\raggedright\arraybackslash}p{0.17\linewidth}p{0.72\linewidth}}
$\Comp_{I,\M}(\Delta)$ & minimum charged local atoms needed to certify rank error at most $\Delta$ on interval $I$;\\
$\RGap_\M$ & general residual-entropy objective, based on conditional entropy after the predictor transcript is known;\\
$\Gap_\M$ & radius-based surrogate for $\RGap_\M$, valid as a lower-bound expression only under rank-spread;\\
$\Rep_\mu(I,h,\Delta)$ & conditional entropy $H(A_I(q)\mid Y_{I,h}(q))$ of the exact answer left for repair;\\
$Y_{I,h}(q)$ & pre-repair transcript: leaf, predictor output, certificate window, and charged side information;\\
rank-spread & condition saying that, after the transcript is observed, many residual ranks remain likely;\\
$B$ & local predictor atom budget; directory, repair-program, and sorted-array storage are accounted for separately;\\
$\kappa_j,\alpha,R_j$ & local hardness, approximation exponent, and rank diameter of leaf $j$.
\end{tabular}
\end{center}

\begin{center}
\fbox{\begin{minipage}{0.90\linewidth}
\textbf{Main theorem intuition within this architecture.}
Routing spends comparisons to identify a leaf.  The local predictor and its finite transcript then reveal some information about the exact rank answer.  Whatever answer entropy remains must be paid by exact repair comparisons.  The upper-bound construction matches this decomposition by using an alphabetic directory plus a transcript-indexed repair program.  That program is outside the local atom budget $B$ and can be large, so the theorem is an architecture-level accounting law rather than a compact-space implementation guarantee.
\end{minipage}}
\end{center}

\subsection*{Running example}

Let $S=\{10,20,\ldots,80\}$ and split the universe into two leaves
$I_1=(-\infty,45]$ and $I_2=(45,\infty)$.  Suppose the workload masses are
$p_1=3/4$ and $p_2=1/4$.  If $I_1$ stores a one-atom affine predictor with
certified radius $\Delta_1=1$ and $I_2$ uses the empty predictor with radius
$\Delta_2=4$, then the radius surrogate is
\[
  \Gap\approx H(3/4,1/4)+\frac34\log 2+\frac14\log 5.
\]
The residual functional can be smaller.  If the transcript value in $I_2$
always identifies that the answer is rank $6$, then
$\Rep_\mu(I_2,h_2,\Delta_2)=0$ although $\Delta_2=4$.  If instead the same
transcript leaves ranks $5,6,7,8$ nearly uniform, then
$\Rep_\mu(I_2,h_2,\Delta_2)\approx\log 4$ and the rank-spread corollary
recovers the radius term.  The atom-allocation problem is exactly the choice
between spending another local atom on $I_1$, where most queries arrive, or on
$I_2$, where the window is wider but the query mass may be small.
In this example, $\Comp_{I_1,\M}(1)=1$ records the one counted segment on
$I_1$, while $\Comp_{I_2,\M}(4)=0$ records that full repair with the empty
predictor is allowed.  $\RGap_\M$ uses the actual residual entropy values
$\Rep_\mu$, whereas $\Gap_\M$ replaces them by $\log(1+\Delta_j)$ only when
rank-spread makes that replacement valid.

\section{Related Work}
\label{sec:related}

\subsection{Classical ordered search, succinctness, and dynamic lower bounds}

The study of ordered search predates learned indexing by decades.  Optimal and nearly optimal binary search trees show that, when the access distribution is known, expected search time is governed by entropy up to additive constants \cite{knuth1971obst,mehlhorn1975obst}.  Biased search trees and splay trees further develop the idea that non-uniform access patterns should be reflected in the shape of the data structure \cite{bent1985biased,sleator1985splay}.  These structures are distribution-sensitive, but they do not exploit geometric regularity of the rank function.  In our language, they optimize routing entropy while setting local repair to zero by storing exact branching information at every decision.

The predecessor problem supplies the complementary lower-bound tradition.  Yao's cell-probe model, Ajtai's lower bound, the Beame-Fich bounds, and the Patrascu-Thorup time-space tradeoffs identify barriers for predecessor search under word-level memory access \cite{yao1977minimax,ajtai1988predecessors,beame2002predecessor,patrascu2006predecessor,patrascu2007randomization}.  These results are worst-case and space-sensitive.  They do not directly answer the workload-conditioned question for a fixed approximation profile.  Our lower bound uses a more explicit learned-index interface and is sharper in its dependence on the instance.

Succinct data structures and rank/select dictionaries provide another relevant axis.  Jacobson's thesis and subsequent indexable dictionaries show how close one can get to information-theoretic space while supporting constant-time operations on bitvectors and related objects \cite{jacobson1989succinct,raman2007fid,patrascu2010succinct}.  Learned rank/select structures reinterpret this compression problem through prediction, using models to represent regular rank functions more compactly \cite{boffa2022rankselect}.  Our model is compatible with this perspective: the approximation profile $\Comp_{I,\M}(\Delta)$ can be viewed as a local compressed representation cost.

Dynamic data structures bring a further obstacle.  Overmars's logarithmic method gives a general dynamization pattern by maintaining exponentially growing static structures \cite{overmars1983design}.  Fredman and Saks established foundational dynamic cell-probe lower bounds, and later work sharpened logarithmic barriers for partial sums and related primitives \cite{fredman1989cellprobe,patrascu2006loglb}.  These lower bounds matter for learned indexes because exact dynamic rank is still dynamic prefix counting, no matter how good the prediction layer is.  Our dynamic section separates what learning can improve in the static search component from what exact dynamic maintenance cannot avoid.

\subsection{Learned indexes: systems evidence and emerging theory}

The systems line begins with the observation that an index can be learned as a function from keys to positions \cite{kraska2018case}.  FITing-Tree and RadixSpline reduce this idea to efficient piecewise-linear approximations \cite{galakatos2019fiting,kipf2020radixspline}.  The PGM-index gives a compressed, error-bounded, and fully dynamic learned index with worst-case guarantees derived from piecewise-linear approximation \cite{ferragina2020pgm}.  ALEX and LIPP explore update-aware layouts and adaptive placement strategies \cite{ding2020alex,wu2021lipp}.  These systems motivate our routed piecewise abstraction: a model gives a coarse position, while a corrective mechanism makes the answer exact.

Benchmarking work has shown that learned indexes are not uniformly superior.  Their performance depends on the data distribution, model hierarchy, error bound, last-mile search procedure, page layout, and update workload \cite{marcus2021benchmarking,maltry2022critical,wongkham2022ready,liu2025pgmpp}.  Disk-resident settings introduce an additional objective in which the relevant repair cost is not just comparisons but block transfers \cite{lan2023disk,zhang2024disk}.  Robustness studies show that dynamic learned indexes may suffer severe degradation under adversarial or skewed update sequences \cite{yang2024aca,luo2025robust}.  These observations support a central premise of this paper: the right theory must expose the interaction between prediction error, workload mass, local difficulty, and exact repair.

The theoretical learned-index literature is now catching up.  Ferragina, Lillo, and Vinciguerra relate the effectiveness of learned indexes to properties of the input distribution \cite{ferragina2020effective}.  Zeighami and Shahabi prove distribution-dependent sub-logarithmic expected query time and near-linear-space constant expected time under mild assumptions \cite{zeighami2023distribution}.  Croquevielle, Yang, Liang, Hadian, and Heinis strengthen the positive theory by obtaining constant expected query time with linear space under weaker probabilistic assumptions and by introducing an information-theoretic statistical-complexity measure \cite{croquevielle2025constant}.  Croquevielle, Sokolovskii, and Heinis develop approximation-theoretic lower bounds for learned indexes, parameterized by the model class \cite{croquevielle2026lower}.

Our result is orthogonal to those lines in three concrete ways.  First, PGM-style theory gives worst-case error-bounded segment guarantees and dynamic compressed indexes; it does not price a workload-conditioned routing distribution together with the residual answer entropy left after the segment transcript.  Second, the positive distributional theorems of Zeighami--Shahabi and Croquevielle et al. identify broad regimes with small expected query time; they do not fix a particular set $S$, workload $\mu$, atom budget $B$, and accounting convention, then ask how the budget should be split across leaves.  Third, approximation-theoretic lower bounds show that a model class may need many pieces on hard rank curves; they do not by themselves combine that local obstruction with query mass and exact-repair entropy.  The contribution here is the accounting equation that couples these quantities for the routed, certified-repair architecture.

There is also a growing line of learned data structures beyond predecessor search, including learned range-minimum queries and learned compressed rank/select structures \cite{ferragina2025flrmq,boffa2022rankselect}.  These works reinforce that learned data structures should not be treated as a single systems trick.  They are a broader attempt to replace uniform worst-case indexing by instance-aware compressed representations.  The present paper contributes a residual-entropy lens to that program.

\section{Formal Model and Profile Complexity}
\label{sec:model}

Let $S=\{x_1<\cdots<x_n\}$ be an ordered set of keys from a totally ordered universe $\U$.  A predecessor query on key $q\in\U$ returns
\[
  \pred_S(q)=\max\{x_i\in S:x_i\le q\},
\]
with the usual sentinel convention when no predecessor exists.  A rank query returns
\[
  \rank_S(q)=|\{x_i\in S:x_i\le q\}|.
\]
The two problems are interchangeable up to standard constant-factor transformations when the structure has access to the sorted key array.  We state most definitions for rank prediction because this is the natural language of learned indexes.

A workload is a probability measure $\mu$ over query keys.  For a measurable interval $I\subseteq\U$, let $\mu(I)$ be the probability that a query lies in $I$.  Throughout the paper logarithms are base two.  Terms with $p\log(1/p)$ are interpreted as zero when $p=0$.  Partitions may contain zero-mass intervals, but such intervals never improve the objective and may be removed.  For clarity, the main statements assume positive leaf masses.

\subsection{Computational and encoding conventions}

The static results are stated in a comparison/RAM hybrid model.  The directory routing phase is an ordered binary decision tree, so each routing step is a comparison against a key or interval boundary.  The local prediction phase may use arithmetic operations on a constant number of machine words per model atom.  Exactness is restored by ordered comparisons made after prediction.  This is the abstraction used by many learned-index analyses: a predictor is useful only when it is cheap to evaluate, and exactness is restored by a last-mile ordered search.

The model deliberately separates two resources.  The atom budget $B$ is a \emph{local approximation budget}: it counts the predictive records that reduce certified repair uncertainty, such as segment records, coefficients, knots, local radix records, and charged exceptions.  It is not the implementation's total byte budget.  Directory topology, directory keys, the sorted key array, and any finite table mapping a pre-repair transcript to a repair comparison program are outside $B$ and must be reported as routing/repair comparisons or as separate byte-level artifact costs.  This separation is not a loophole: it prevents the theorem from pretending to be a byte-space lower bound while still isolating the tradeoff the paper studies, namely local predictive information versus residual exact-search uncertainty.  A one-resource byte theorem would need an additional directory-and-repair-program space term.

Two admissibility conventions keep the local predictor from hiding an exact index.  First, every piece of information used by a local predictor is charged as atoms; arbitrary tables of exceptional ranks are not free.  Second, after a leaf predictor returns a certified rank interval of radius $\Delta$, the repair algorithm may compare the query with stored keys inside that interval, but it receives no extra oracle revealing the exact rank.  These conventions are standard in cell-probe and comparison lower bounds: computation can be powerful, but stored information and ordered distinctions must be accounted for.

For dynamic lower bounds we use the cell-probe model with $w$-bit cells.  The cell-probe model charges only memory accesses and gives computation for free.  A lower bound in this model therefore applies to any RAM implementation.  We use the model only to transfer classical dynamic subset-rank lower bounds to exact dynamic learned rank structures.

\subsection{Atomized local model families}

The phrase ``model atom'' is too weak unless it fixes what information is being budgeted.  We use the following abstract encoding.  It is deliberately restrictive enough to support lower bounds, but broad enough to include the segment-based learned indexes used in practice.

\begin{definition}[Atomized model family]
An atomized model family $\M$ is a sequence $(\M_a)_{a\ge 0}$ of local rank predictors together with an encoding map.  A predictor in $\M_a$ has a code consisting of at most $a$ constant-size records, each record using $O(w)$ bits on a word RAM with $w=\Theta(\log n)$ unless otherwise stated.  Evaluation may perform arbitrary computation on the decoded records and the query key, but every stored coefficient, breakpoint, exception, pointer to an auxiliary table, or payload used by the local predictor is charged to one of the $a$ records.  The classes are monotone, $\M_a\subseteq\M_{a+1}$, and $\M_0$ contains only the empty predictor, which may be used with the trivial radius of the leaf.
\end{definition}

The encoding convention is part of the theorem, not an implementation detail.  Atoms are finite word records; real-valued coefficients are understood as finite encodings with enough precision to certify the stated window, and any additional precision must be charged.  A local model may compute an arbitrary deterministic function of the query and its charged records, but it may not consult an uncharged table, unbounded-precision constant, hidden hash map, or payload that distinguishes individual ranks.  If a design stores exceptions, learned buckets with exact offsets, or a secondary table of positions, those records are atoms.  These restrictions are what prevent the predictor from storing an exact index in disguise.

The definition charges information, not syntax.  For an affine piecewise-linear family, one atom is a constant-size segment record containing its domain boundary and two coefficients.  For a spline or RadixSpline-style family, one atom is a knot or radix table entry together with its stored position information.  For a bucketed family, one atom is a bucket descriptor.  An exception table is allowed, but each exception entry is an atom; otherwise the lower bound would be vacuous because exact ranks could be hidden in the model layer.

\begin{definition}[Local rank diameter]
For a contiguous interval $I\subseteq\U$, define
\[
  R_I=\diam_S(I)=1+\max_{q,q'\in I}|\rank_S(q)-\rank_S(q')|.
\]
The trivial empty predictor has certified radius at most $R_I$ on $I$.
\end{definition}

\begin{definition}[Local approximation profile]
Let $I\subseteq\U$ be a contiguous key interval and let $\Delta\ge 0$ be an integer.  The local approximation profile of $I$ with respect to $\M$ is
\[
\begin{aligned}
  \Comp_{I,\M}(\Delta)=\min\Bigl\{a:\;&\text{there exists }h\in\M_a\text{ such that}\\
  &|h(q)-\rank_S(q)|\le \Delta \text{ for every }q\in I\Bigr\}.
\end{aligned}
\]
If no such predictor exists, the value is $+\infty$.  By convention $\Comp_{I,\M}(R_I)=0$ because the empty predictor and full repair are always allowed.
\end{definition}

The function $\Comp_{I,\M}(\Delta)$ is nonincreasing in $\Delta$ and integer-valued.  It is the only point where the local modeling family enters the static theorems.  The routing and repair arguments therefore apply unchanged to linear segments, bucket models, splines, or other families once their certified local profiles have been specified.

\subsection{Computability and certifiability of profiles}

The parameter $\Comp_{I,\M}(\Delta)$ should not be read as a black box oracle.  In an application one usually has a certified sandwich
\[
  L_I(\Delta)\le \Comp_{I,\M}(\Delta)\le U_I(\Delta),
\]
where $U_I$ is produced by an explicit construction and $L_I$ by an obstruction argument.  All lower bounds in this paper may use $L_I$, all upper bounds may use $U_I$, and a constant-factor sandwich changes the final bounds only by constant factors.

For piecewise-linear predictors, $U_I(\Delta)$ is obtained by the minimum or near-minimum number of error-bounded segments covering the local rank curve.  This is precisely the segmentation primitive exploited by PGM-like indexes.  Lower certificates can be obtained from disjoint subintervals whose rank curves cannot be approximated by a single admissible segment at error $\Delta$; more generally they come from packing, curvature, or metric-entropy arguments.  For bucket models, $U_I$ is a bucketing scheme and $L_I$ follows from the number of buckets needed to separate rank jumps.  For parametric families, $L_I$ may be obtained from pseudo-dimension, covering-number, or approximation-width bounds.  Thus exact profile computation is not required for the theorem to be meaningful; certified profile bounds are enough.

\begin{definition}[Certified profile sandwich]
A profile sandwich for a family $\M$ on a collection of intervals is a pair of integer-valued functions $(L_I,U_I)$ satisfying $L_I(\Delta)\le \Comp_{I,\M}(\Delta)\le U_I(\Delta)$ for all certified radii.  The sandwich has quality $\gamma$ on a radius range if $U_I(\Delta)\le \gamma\max\{1,L_I(\Delta)\}$ throughout that range.
\end{definition}

With a profile sandwich, one obtains two computable functionals: a lower functional using $L_I$ and an upper functional using $U_I$.  If the sandwich quality is bounded and the relevant optimum avoids only zero-mass intervals, the two functionals differ by at most the corresponding constant factors and the same shadow-price calculus applies.

\subsection{A concrete profile: error-bounded linear segments}

We now spell out one model family for which the profile is not an oracle.  Let $Q_I=(q_1<\cdots<q_N)$ be the ordered set of certified control keys in a leaf interval $I$.  For exact predecessor and rank certificates, $Q_I$ contains the relevant stored keys and gap endpoints at which the rank function can change; for a finite benchmark trace, the same definition can be applied to the trace keys but then the result is a trace certificate rather than a certificate for every query in $I$.  Let
\[
  f_I(t)=\rank_S(q_t),\qquad t=1,\ldots,N,
\]
be the certified local rank curve on those control keys.  A block $[a,b]\subseteq[N]$ is $\Delta$-linear if there exist finite-precision encoded coefficients $u,v$ such that
\[
  |u q_t+v-f_I(t)|\le \Delta\qquad \text{for every }t\in[a,b].
\]
Let $N_I^{\Lin}(\Delta)$ be the minimum number of contiguous $\Delta$-linear blocks whose union is $[N]$.

\begin{proposition}[Exact profile for counted affine segments]
\label{prop:linear-profile}
Let $\M^{\Lin}_a$ be the family of local predictors consisting of at most $a$ contiguous affine segments, where each segment record stores its right endpoint and two coefficients and is charged as one atom.  Then, up to the constant convention for endpoint records,
\[
  \Comp_{I,\M^{\Lin}}(\Delta)=N_I^{\Lin}(\Delta).
\]
Moreover, $N_I^{\Lin}(\Delta)$ is computable by the dynamic program
\[
  D[b]=1+\min\{D[a-1]:[a,b]\text{ is }\Delta\text{-linear}\},\qquad D[0]=0,
\]
provided feasibility of a block is available.
\end{proposition}

\begin{proof}
Every cover by $N_I^{\Lin}(\Delta)$ feasible affine blocks gives a predictor with that many counted segment atoms and certified error at most $\Delta$, proving the upper bound.  Conversely, any predictor in $\M_a^{\Lin}$ partitions the local domain into at most $a$ contiguous affine pieces, and each piece must be $\Delta$-linear under the certificate.  Hence those pieces form a feasible cover, so $a\ge N_I^{\Lin}(\Delta)$.  The dynamic program is the standard shortest-path formulation on prefix endpoints: an edge $(a-1,b)$ exists exactly when $[a,b]$ is feasible.  For one-dimensional affine $\ell_\infty$ error, feasibility can be certified by intersecting the slope intervals induced by the inequalities
\[
  f_I(t)-\Delta\le u q_t+v\le f_I(t)+\Delta,
\]
so the certificate is checkable and not empirical.
\end{proof}

This proposition is the concrete case to keep in mind.  For a PGM-style index, an atom is an error-bounded segment and $\Delta$ is the usual last-mile search radius.  For a RadixSpline-style index, an atom is a spline or radix record and the same profile logic applies after charging the stored dispatch records.  In the experimental artifact we also report sampled profile curves, because they are useful diagnostics on large traces.  Those sampled curves are labeled empirical estimates; only the all-control-key certificate gives the theorem-level guarantee.

\begin{center}
\scriptsize
\begin{tabular}{lll}
\hline
Component & Atom charge & Outside $B$\\
\hline
PGM/FITing seg. & slope/intercept/error & directory + sorted array\\
RadixSpline rec. & spline predictor rec. & radix table + sorted array\\
Exceptions & charged predictor payload & exact-rank payload/caches\\
Shadow prototype & selected segments & route dir + repair comps\\
\hline
\end{tabular}
\end{center}

This mapping is why the atom model is not chosen after the fact to force the theorem.  It isolates the common predictive records whose purpose is to reduce the certified repair window.  Additional dispatch machinery is allowed, but it must be visible either in the transcript, in routing comparisons, or in a separate byte count.

\subsection{Routed piecewise learned indexes}

\begin{center}
\fbox{\begin{minipage}{0.92\linewidth}
\phantomsection\label{conv:accounting}
\textbf{Accounting convention used by the theorem.}
$B$ counts only charged local predictor atoms.  Directory topology, directory keys, the sorted key array, and transcript-to-repair-program tables are outside $B$.  They are not ignored: they are reported as routing/repair comparisons or as explicit byte-level artifact costs.  Thus \cref{thm:residual-law} is a local-predictor budget theorem, not a byte-space lower bound.
\end{minipage}}
\end{center}

\begin{definition}[Routed piecewise learned index]
A routed piecewise learned index for $(S,\mu)$ consists of:
\begin{enumerate}[label=(\roman*),leftmargin=1.5em]
  \item an ordered partition $\Pi=\{I_1,\ldots,I_m\}$ of $\U$ into key-aligned contiguous intervals;
  \item an ordered binary directory tree $\Dir$ whose leaves are exactly $I_1,\ldots,I_m$ in left-to-right order;
  \item for each positive-mass leaf $I_j$, either the empty predictor with radius $R_{I_j}$ or a local predictor $h_j\in\M_{a_j}$ with certified worst-case rank error at most $\Delta_j$;
  \item an exact local repair procedure that, after reaching $I_j$, returns the predecessor or rank by comparisons inside the certified window.
\end{enumerate}
The atom budget is $\sum_j a_j$.  Directory keys, tree topology, and sorted data are not counted as model atoms; they are charged through routing and repair comparisons.
\end{definition}

The atom budget measures only the approximation layer identified in the boxed convention.  A system evaluation must still report directory bytes, sorted-array assumptions, and repair-program storage.  None of the lower bounds become weaker if those structures are also charged as memory; the resulting statement is simply a two-resource theorem rather than the local predictor-budget theorem studied here.

The ordered binary routing condition fixes the theorem interface.  In a PGM-like interpretation, an atom is a certified linear segment, the repair radius is the PGM error parameter $\varepsilon$, routing is the ordered search among segments or upper-level recursive models, and repair is binary or exponential search inside $\pm\varepsilon$ positions.  In a RadixSpline-like interpretation, atoms are spline and radix records, routing uses the charged radix/spline directory plus an ordered local search, and repair resolves the bracket left by interpolation.  If a router uses arithmetic dispatch, a neural layer, hashing, or exact-position payloads, the corresponding dispatch records or payloads must be added to the charged transcript or to a separate memory term before applying the theorem.

\subsection{Leaf entropy and certified repair}

For a partition $\Pi=\{I_1,\ldots,I_m\}$, write
\[
  p_j=\mu(I_j).
\]
The leaf entropy of the partition is
\[
  H_\mu(\Pi)=\sum_{j=1}^m p_j\log\frac{1}{p_j}.
\]
If a query routed to leaf $I_j$ traverses depth $d_j$ in the directory tree and then performs exact correction inside a window of radius $\Delta_j$, the natural total cost is
\[
  d_j+\Theta(\log(1+\Delta_j)).
\]
The routing cost is information-theoretic; the repair cost comes from exact ordered search inside a candidate set of size $\Theta(1+\Delta_j)$.

\subsection{Residual entropy and rank spread}

A certified radius is a worst-case promise, not an average-case lower bound.  The repair algorithm does not enter the last-mile search knowing only the leaf.  It also knows the output of the local predictor, the certified interval, and any deterministic pre-repair information computed from charged atoms.  The residual uncertainty must be measured after this information is revealed.

\begin{definition}[Pre-repair transcript]
Fix a leaf interval $I$, a charged local predictor $h$, and its certificate rule.  For a query $q\in I$, the pre-repair transcript
\[
  Y_{I,h}(q)
\]
is the finite information available to the exact repair phase before its first comparison with a stored key: the leaf identity, the encoded predictor identity, the value $h(q)$, the certified rank interval or window endpoints derived from $h(q)$, and any deterministic side information computed from the charged local atoms.  It does not include the outcomes of comparisons between $q$ and stored keys.  The exact answer is denoted
\[
  A_I(q)\in\{0,\ldots,n\}.
\]
\end{definition}

\begin{center}
\fbox{\begin{minipage}{0.90\linewidth}
\textbf{Transcript example.}
For a segment leaf, a typical finite word-RAM transcript is
\[
\begin{aligned}
Y_{I,h}(q)=(&\mathrm{leaf}=17,\ \mathrm{segment}=17,\ \widehat r=1048576,\\
&[\ell,u],\ \Delta,\ \mathrm{encoded\ local\ coefficients}).
\end{aligned}
\]
It may include rounded predictor output and certified window endpoints.  It does not include answers to probes such as ``is $q<S[k]$?''; those ordered comparisons are the repair cost lower bounded below.
\end{minipage}}
\end{center}

This definition follows the comparison lower-bound convention.  In an implemented word-RAM structure the transcript is a finite tuple of words: encoded coefficients, interval identifiers, rounded predictor output, certified endpoints, and finite side information.  This finite version is the primary implementation model.  If one writes a mathematical predictor with real-valued output, $Y_{I,h}$ can instead be treated as a measurable random variable and
\[
  H(A\mid Y)=\int H(P_{A\mid Y=y})\,dP_Y(y),
\]
using a regular conditional distribution of the finite answer variable $A$.  The finite implementation is obtained by the word encoding of the predictor output and certificate; any extra precision or lookup table used to refine the transcript is charged as atoms or as explicit non-atom memory.  Ordered facts learned by comparing $q$ with stored keys are not part of the transcript; they are exactly the information whose cost is being lower bounded.

\begin{definition}[Conditional leaf repair entropy]
Fix a leaf interval $I$, a predictor $h$, a certified radius $\Delta$, and the conditional workload $\mu_I$.  The conditional repair entropy is
\[
  \Rep_{\mu}(I,h,\Delta)=H\!\left(A_I(q)\mid Y_{I,h}(q)\right),
  \qquad q\sim\mu_I.
\]
Equivalently,
\[
  \Rep_{\mu}(I,h,\Delta)
  =\sum_y \Prb[Y_{I,h}=y]\,
     H\!\left(A_I(q)\mid Y_{I,h}(q)=y\right).
\]
For a measurable non-finite transcript this sum is replaced by the integral above.  Since $A_I$ has at most $n+1$ values, the conditional entropy is always finite and the finite-transcript formulas are recovered by discretization.
\end{definition}

The conditioning is essential.  If the predictor is exact, then $Y_{I,h}$ already determines $A_I$ and the repair entropy is zero.  If the predictor returns a large window but the workload concentrates on one answer inside that window, the conditional entropy is also small.  If, after the same transcript value, many ranks remain plausible with non-negligible conditional probability, the entropy is logarithmic in the number of plausible ranks.

\begin{definition}[Residual-entropy functional]
For an atomized family $\M$, define
\begin{equation}
\label{eq:rgap-def}
\begin{aligned}
  \RGap_{\M}(S,\mu,B)=
  \inf_{\substack{\Pi=\{I_1,\ldots,I_m\}\\ a_j,h_j,\Delta_j\\ \sum_j a_j\le B}}
  \Biggl(
    H_\mu(\Pi)+\sum_{j=1}^{m}p_j\Rep_\mu(I_j,h_j,\Delta_j)
  \Biggr),
\end{aligned}
\end{equation}
where $p_j=\mu(I_j)$, $h_j\in\M_{a_j}$, and $h_j$ has certified worst-case rank error at most $\Delta_j$ on $I_j$.  The infimum ranges over ordered partitions, charged local predictors, and their certificates.
\end{definition}

This is the primary instance parameter.  It is not a radius parameter: two leaves with the same certified radius may have different conditional residual entropies because their predictor transcripts may leave different answer distributions.

\begin{definition}[Rank-spread condition]
For constants $c_0,c_1>0$, a positive-mass leaf $I$ with predictor $h$ is $(c_0,c_1)$-rank-spread if, for every positive-probability transcript value $y$ of $Y_{I,h}$, the conditional residual exact-answer distribution contains a subset of at least $c_0(1+\Delta_y)$ candidate ranks, each of conditional probability at least $c_1/(1+\Delta_y)$, where $\Delta_y$ is the certified radius of the window encoded by transcript $y$.  For fixed-radius leaves this reduces to the requirement that, after the predictor output and certificate are known, $\Theta(1+\Delta)$ ranks remain conditionally plausible.  A structure is rank-spread if every positive-mass leaf satisfies this condition for its realized predictor and transcript distribution.
\end{definition}

Rank-spread is not a technical afterthought; it is exactly the hypothesis that turns conditional residual entropy into a radius lower bound.  It is plausible when conditional queries are approximately uniform over many gaps in a leaf, when a Zipfian workload has had its heaviest ranks separated into singleton or small leaves and the remaining tail is not concentrated on one residual answer, and when range or sensor workloads add local jitter around predicted positions.  It fails for repeated point lookups at a single hot key, for membership checks concentrated in one gap, and for cache-filtered workloads in which an upper layer removes all but a few residual outcomes.  In those failing cases the theorem below remains true with $\RGap_\M$, while the radius expression below is not a lower bound.

\begin{example}[Rank-spread workloads and failures]
\begin{center}
\scriptsize
\begin{tabular}{lll}
\hline
Workload & Holds when & Fails when\\
\hline
Uniform local & broad rank mass & endpoint-only mass\\
Time/sensor jitter & adjacent ranks likely & cache pre-resolves answer\\
Zipf tail & hot keys isolated & head dominates transcript\\
Sparse OSM locality & many plausible gaps & one gap repeated\\
\hline
\end{tabular}
\end{center}
These failures are not pathologies; they are exactly the cases where $\RGap_\M$ is smaller than the radius surrogate.
\end{example}

\begin{definition}[Empirical rank-spread diagnostic]
Fix a benchmark implementation and a finite trace.  Let $Z(q)$ be a documented coarsening of the finite pre-repair transcript, such as the leaf id together with a bucketed repair-window size.  For each positive-mass bucket $z$, let $\widehat P_{A\mid Z=z}$ be the empirical residual-rank histogram and let $\widehat\Delta_z$ be the median certified radius in that bucket.  The empirical entropy ratio is
\[
  \widehat\rho
  =
  \frac{\sum_z \widehat P[Z=z]\,H(\widehat P_{A\mid Z=z})}
       {\sum_z \widehat P[Z=z]\,\log(1+\widehat\Delta_z)}.
\]
The accompanying support statistic is the average number of residual ranks with empirical probability at least $1/(4(1+\widehat\Delta_z))$ inside their bucket.
\end{definition}

This diagnostic is not a theorem assumption and does not equal $H(A\mid Y)$ unless the coarsening $Z$ is the actual implementation transcript.  It is a trace-level test for whether the radius surrogate is plausible.  High ratios and broad support indicate rank-spread behavior; low ratios indicate that the residual-entropy objective should be used directly.

\begin{lemma}[Radius is only a spread surrogate]
\label{lem:spread-surrogate}
For every certified leaf,
\[
  0\le \Rep_\mu(I,h,\Delta)\le O(\log(1+\Delta)).
\]
If the leaf is rank-spread, then
\[
  \Rep_\mu(I,h,\Delta)=\Omega(\log(1+\Delta)).
\]
Both inequalities are tight up to constants.
\end{lemma}

\begin{proof}
Condition on a transcript value $y$.  The certified window for that transcript contains at most $O(1+\Delta_y)$ possible ranks, so its answer entropy is at most $O(\log(1+\Delta_y))$; averaging over $y$ gives the upper bound for fixed-radius leaves and the analogous variable-radius statement.  Under rank-spread, each positive-probability transcript leaves $\Theta(1+\Delta_y)$ answers with probability $\Omega(1/(1+\Delta_y))$, so the conditional entropy is $\Omega(\log(1+\Delta_y))$.  The uniform residual distribution gives tightness at the upper end.  A workload supported on a single residual answer gives entropy zero even when the certified window is large.
\end{proof}

\begin{definition}[Rank-spread radius functional]
For a model class $\M$, ordered set $S$, workload $\mu$, and atom budget $B$, define
\begin{equation}
\label{eq:gap-def}
\Gap_{\M}(S,\mu,B)=
\min_{\substack{\Pi=\{I_1,\ldots,I_m\}\\ \Delta_1,\ldots,\Delta_m\ge 0\\ \sum_{j=1}^{m}\Comp_{I_j,\M}(\Delta_j)\le B}}
\sum_{j=1}^{m}p_j\left(\log\frac{1}{p_j}+\log(1+\Delta_j)\right),
\end{equation}
where $p_j=\mu(I_j)$.
\end{definition}

The quantity $\Gap_\M$ is a \emph{rank-spread surrogate} for $\RGap_\M$, not the general lower-bound parameter.  It removes the conditional transcript distribution inside each window and replaces it by the certified radius.  This replacement is mathematically valid only when residual answers remain sufficiently spread after the predictor output has been observed.

\begin{remark}[Sanity checks for the surrogate]
If every leaf predicts exactly, so that $\Delta_j=0$, then \cref{eq:gap-def} collapses to $H_\mu(\Pi)$, recovering alphabetic-search entropy.  If the partition has one interval and only the trivial radius $\Delta=\Theta(n)$ is certified, the repair term is $\Theta(\log n)$.  These checks concern the surrogate; the residual functional can be strictly smaller when the workload has low conditional residual entropy.
\end{remark}

\section{Shadow Prices}
\label{sec:calculus}

The residual functional is the general object.  The radius functional $\Gap_\M$ is the tractable rank-spread surrogate.  This section develops the Lagrangian calculus for that surrogate.  The calculus has two roles.  First, it provides lower-bound certificates for every feasible radius allocation under rank-spread.  Second, it yields the water-filling law that determines which intervals deserve accurate predictors.

For an interval $I$ with conditional mass parameter $p>0$, define the leaf objective
\[
  F_{I,p}(\Delta)=p\log(1+\Delta),
\]
and let the local storage cost be $C_I(\Delta)=\Comp_{I,\M}(\Delta)$.  For a multiplier $\lambda\ge 0$, define the \emph{leaf shadow envelope}
\begin{equation}
\label{eq:shadow-envelope}
  \Psi_{I,p}(\lambda)=\min_{\Delta\ge 0}\{F_{I,p}(\Delta)+\lambda C_I(\Delta)\}.
\end{equation}
The value $\lambda$ is the exchange rate between one model atom and one bit of expected repair information.

\begin{definition}[Repair-entropy potential]
For a partition $\Pi=\{I_1,\ldots,I_m\}$, local radii $\Delta_1,\ldots,\Delta_m$, and multiplier $\lambda\ge 0$, define
\begin{equation}
\label{eq:phi-def}
  \Phi_\lambda(\Pi,\Delta)=
  H_\mu(\Pi)+\sum_{j=1}^{m}p_j\log(1+\Delta_j)
  +\lambda\sum_{j=1}^{m}\Comp_{I_j,\M}(\Delta_j).
\end{equation}
\end{definition}

\begin{proposition}[Shadow-price separability]
\label{prop:separable}
Fix a partition $\Pi=\{I_1,\ldots,I_m\}$ and $\lambda\ge 0$.  Minimizing $\Phi_\lambda(\Pi,\Delta)$ over integer radii separates leafwise:
\[
\argmin_{\Delta_1,\ldots,\Delta_m\ge 0}\Phi_\lambda(\Pi,\Delta)=
\prod_{j=1}^{m}\argmin_{\Delta\ge 0}\left\{p_j\log(1+\Delta)+\lambda\Comp_{I_j,\M}(\Delta)\right\}.
\]
Moreover, for every budget $B$ and every $\lambda\ge 0$,
\begin{equation}
\label{eq:dual-cert}
\min_{\substack{\Delta_j\ge 0\\ \sum_j\Comp_{I_j,\M}(\Delta_j)\le B}}
\sum_{j=1}^{m}p_j\log(1+\Delta_j)
\ge
\sum_{j=1}^{m}\Psi_{I_j,p_j}(\lambda)-\lambda B.
\end{equation}
\end{proposition}

\begin{proof}
The first claim follows because $H_\mu(\Pi)$ is independent of $\Delta$ and the remaining Lagrangian is a sum of one-dimensional objectives.  For \eqref{eq:dual-cert}, every feasible radius vector satisfies
\[
  \sum_jp_j\log(1+\Delta_j)
  \ge \sum_j\bigl(p_j\log(1+\Delta_j)+\lambda\Comp_{I_j,\M}(\Delta_j)\bigr)-\lambda B.
\]
Minimizing the right-hand side coordinatewise gives the stated bound.
\end{proof}

The inequality \eqref{eq:dual-cert} is useful because it is a certificate: to prove that an allocation cannot beat a claimed value, it is enough to exhibit a multiplier $\lambda$.  When the discrete profiles are replaced by their lower convex envelopes, the reverse inequality follows from standard convex duality.  Thus the shadow envelope is not merely an analysis trick; it is the dual form of the atom-allocation problem.

\begin{definition}[Entropy deficiency]
\label{def:deficiency}
For a fixed partition $\Pi$ and budget $B$, define the repair deficiency
\[
  \Def_{\M}(\Pi,B)=
  \min_{\substack{\Delta_j\ge 0\\ \sum_j\Comp_{I_j,\M}(\Delta_j)\le B}}
  \sum_{j=1}^{m}p_j\log(1+\Delta_j).
\]
Thus $\Gap_\M(S,\mu,B)=\min_\Pi\{H_\mu(\Pi)+\Def_\M(\Pi,B)\}$.
\end{definition}

The name is deliberate.  $H_\mu(\Pi)$ is the information needed to choose a local model; $\Def_\M(\Pi,B)$ is the information not removed by the model budget and therefore still paid by exact repair.  A learned index is effective precisely when the budget makes this deficiency small on high-mass intervals.

\begin{definition}[Normalized discrete power-law profile]
\label{def:powerlaw}
Let $R_j=\diam_S(I_j)$.  A fixed partition $\Pi=\{I_1,\ldots,I_m\}$ has an $(\alpha,\kappa,R)$-regular profile for $\M$ if there are constants $a,b,\alpha>0$ and local hardness values $\kappa_j>0$ such that for every leaf and every integer $0\le \Delta\le R_j$,
\begin{equation}
\label{eq:powerlaw}
  a\bigl(1+\kappa_j(1+\Delta)^{-\alpha}\bigr)
  \le 1+\Comp_{I_j,\M}(\Delta)
  \le b\bigl(1+\kappa_j(1+\Delta)^{-\alpha}\bigr),
\end{equation}
and $\Comp_{I_j,\M}(R_j)=0$ up to the empty-predictor convention.  The additive one inside $1+\Comp$ is the atom floor; the cap at $R_j$ is the full-repair boundary.
\end{definition}

This is the discrete version used in the theorem.  It avoids the impossible conclusion that fewer than one atom certifies a nontrivial local predictor.  It also records the boundary case in which the model stores no local predictor and the repair window spans the entire leaf.

\begin{proposition}[Water-filling law]
\label{prop:local-waterfilling}
Assume \cref{def:powerlaw}.  Fix a partition $\Pi$ and multiplier $\lambda>0$.  Every asymptotic minimizer of the shadow objective satisfies
\[
  1+\Delta_j(\lambda)=\Theta\left(
  \min\left\{1+R_j,\max\left\{1,\left(\frac{\lambda\kappa_j}{p_j}\right)^{1/\alpha}\right\}\right\}\right),
\]
with constants depending only on $a,b,\alpha$.
\end{proposition}

\begin{proof}
For a non-boundary radius, the local objective is approximated by
\[
  p_j\log(1+\Delta)+\lambda\kappa_j(1+\Delta)^{-\alpha}.
\]
Balancing derivatives gives
\[
  \frac{p_j}{1+\Delta}\asymp \lambda\kappa_j(1+\Delta)^{-\alpha-1},
\]
so $(1+\Delta)^\alpha\asymp \lambda\kappa_j/p_j$.  The lower truncation enforces integrality and the atom floor; the upper truncation enforces the full-repair boundary $\Delta\le R_j$.  Rounding to the nearest integer changes the objective by at most a constant factor because $\log(1+\Delta)$ and the normalized profile vary by constant factors under constant-factor radius perturbations.
\end{proof}

The law is asymmetric in the right way.  Increasing $p_j$ makes a leaf more expensive to leave unrepaired, so its optimal radius decreases.  Increasing $\kappa_j$ makes the same reduction in radius more costly, so the radius increases unless the query mass justifies the atoms.  The relevant local statistic is therefore the hardness-to-mass ratio $\kappa_j/p_j$, not hardness alone.

\begin{definition}[Hard-mass profile]
\label{def:hardmass}
For a power-law partition define
\[
  W_\Pi(t)=\sum_{j:\,\kappa_j/p_j\ge t}p_j,
\]
the query mass lying on leaves whose hardness-to-mass ratio is at least $t$.
\end{definition}

\begin{lemma}[Layer-cake form of repair deficiency]
\label{lem:layercake}
For every fixed power-law partition,
\[
  \sum_{j=1}^{m}p_j\logplus\frac{\kappa_j}{\Beff p_j}
  =\int_{0}^{\infty} W_\Pi(\Beff 2^u)\,du,
\]
where logarithms are base two.
\end{lemma}

\begin{proof}
For each $j$, $\logplus(\kappa_j/(\Beff p_j))=\int_0^\infty \mathbf{1}\{\kappa_j/(\Beff p_j)\ge 2^u\}\,du$.  Multiplying by $p_j$ and interchanging the finite sum with the integral gives the identity.
\end{proof}

The hard-mass profile is the most concise way to state the obstruction.  Learned indexing is not defeated by a few pathological intervals of negligible probability; it is defeated when a non-negligible amount of query mass remains above the budget scale $B$ in hardness-to-mass ratio.

\subsection{Information barriers for routing and repair}

\begin{lemma}[Entropy lower bound for routing]
\label{lem:routing-lb}
Let $\Dir$ be any ordered binary directory tree whose leaves are queried with probabilities $p_1,\ldots,p_m$.  If leaf $j$ has depth $d_j$, then
\[
  \sum_{j=1}^{m}p_jd_j\ge \sum_{j=1}^{m}p_j\log\frac{1}{p_j}=H_\mu(\Pi).
\]
\end{lemma}

\begin{proof}
By Kraft's inequality, $\sum_j2^{-d_j}\le 1$.  Let $q_j=2^{-d_j}$ and $Q=\sum_jq_j$.  If $Q=0$ the statement is trivial, so assume $Q>0$ and set $\widehat q_j=q_j/Q$.  Then $\widehat q$ is a probability distribution and
\[
  \sum_jp_jd_j=\sum_jp_j\log\frac{1}{q_j}
  =\sum_jp_j\log\frac{1}{p_j}+\sum_jp_j\log\frac{p_j}{q_j}.
\]
The second term equals
\[
  \sum_jp_j\log\frac{p_j}{\widehat q_j}+\log\frac{1}{Q}
  =\KL(p\|\widehat q)+\log\frac{1}{Q}.
\]
Both summands are nonnegative: the first by Gibbs' inequality and the second because $Q\le 1$.  Therefore $\sum_jp_jd_j\ge H_\mu(\Pi)$.
\end{proof}

\begin{lemma}[Nearly optimal alphabetic routing]
\label{lem:routing-ub}
For every ordered set of leaf probabilities $p_1,\ldots,p_m$, there exists an ordered binary directory tree with expected routing cost at most $H_\mu(\Pi)+2$.
\end{lemma}

\begin{proof}
The ordered leaves define an alphabetic distribution: every admissible routing tree must keep the leaves in the left-to-right order of the key intervals.  Mehlhorn's construction for alphabetic search trees assigns to each ordered probability $p_j$ a code interval of length proportional to $p_j$ and chooses split keys so that the resulting binary tree has leaf depths satisfying an expected-depth bound
\[
  \sum_jp_jd_j\le H(p_1,\ldots,p_m)+2
\]
\cite{mehlhorn1975obst}.  Applying that construction to the leaf masses gives an ordered directory tree, because the construction never permutes the leaves.  The expected routing cost of this directory is therefore at most $H_\mu(\Pi)+2$.
\end{proof}

\begin{lemma}[Conditional residual-answer lower bound]
\label{lem:residual-entropy}
Condition on reaching a leaf $I$ with predictor $h$ and certificate radius $\Delta$.  Let $Y=Y_{I,h}(q)$ be the pre-repair transcript and let $A=A_I(q)$ be the exact predecessor or rank answer for $q\sim\mu_I$.  Every comparison-based repair procedure correct for all queries in $I$ has expected depth at least
\[
  H(A\mid Y)-O(1).
\]
\end{lemma}

\begin{proof}
Fix a transcript value $Y=y$.  Once $y$ is fixed, the repair algorithm may choose a decision tree specialized to that transcript; the only remaining branches are comparisons between $q$ and stored keys in the certified interval.  The leaves of this tree refine the possible exact answers.  Collapsing all leaves with the same exact answer cannot increase expected depth, and yields a prefix code for the conditional answer distribution $A\mid Y=y$.  Shannon's lower bound for binary prefix codes gives conditional expected depth at least $H(A\mid Y=y)-O(1)$.  Averaging over transcript values gives $H(A\mid Y)-O(1)$.
\end{proof}

\begin{lemma}[Worst-case window lower bound]
\label{lem:repair-lb}
Fix a leaf interval $I$ with promised error radius $\Delta\ge 0$.  In the worst case, exact repair requires $\Omega(\log(1+\Delta))$ comparisons.
\end{lemma}

\begin{proof}
Choose a transcript/window in which $\Theta(1+\Delta)$ ranks are feasible.  An exact decision tree distinguishing these possibilities has at least that many leaves and therefore depth $\Omega(\log(1+\Delta))$.
\end{proof}

\begin{lemma}[Window repair upper bound]
\label{lem:repair-ub}
Fix a leaf interval $I$ with promised error radius $\Delta\ge 0$.  Exact repair can be done in $O(\log(1+\Delta))$ comparisons.
\end{lemma}

\begin{proof}
Fix the transcript produced before repair.  The certificate gives a clipped interval of candidate ranks
\[
  [\ell,u]\subseteq\{0,1,\ldots,n\}
\]
with $u-\ell+1\le c(1+\Delta)$ for an absolute constant $c$; boundary clipping can only shrink the interval.  The predecessor or rank answer is the first position in this interval whose stored key is not smaller than $q$, with the usual sentinel convention at the endpoints.  Standard binary search on $S[\ell],\ldots,S[u]$ identifies this boundary using at most
\[
  \lceil\log_2(u-\ell+2)\rceil=O(\log(1+\Delta))
\]
comparisons.  The returned position is exact because the certificate guarantees that the true answer lies in the interval.
\end{proof}

\begin{lemma}[Rank-spread expected repair]
\label{lem:rank-spread-repair}
If a leaf is rank-spread, then the expected repair cost under the conditional workload is $\Omega(\log(1+\Delta))$ in the fixed-radius case, and $\Omega(\mathbb{E}[\log(1+\Delta_Y)])$ in the transcript-dependent radius case.
\end{lemma}

\begin{proof}
Condition on a positive-probability transcript value $Y=y$.  By rank-spread there is a set $R_y$ of at least $c_0(1+\Delta_y)$ residual exact answers, and each answer in $R_y$ has conditional probability at least $c_1/(1+\Delta_y)$.  Therefore
\[
\begin{aligned}
  H(A\mid Y=y)
  &\ge \sum_{a\in R_y}\Prb[A=a\mid Y=y]\log\frac{1}{\Prb[A=a\mid Y=y]}\\
  &\ge c_0(1+\Delta_y)\frac{c_1}{1+\Delta_y}
        \log\frac{1+\Delta_y}{c_1}\\
  &=\Omega(\log(1+\Delta_y)),
\end{aligned}
\]
where the hidden constant depends only on the fixed spread constants.  Applying \cref{lem:residual-entropy} to the conditional repair tree and averaging over $y$ gives the transcript-dependent lower bound.  If the radius is fixed, then $\Delta_y=\Delta$ for every transcript value and the expectation reduces to $\Omega(\log(1+\Delta))$.
\end{proof}

These lemmas identify the two independent sources of cost.  Routing is constrained by Kraft's inequality; repair is constrained by residual answer entropy.  The next section combines them with the atom-budget constraint.

\section{Static Residual Law}
\label{sec:static}

We now separate the two mathematical statements that were often conflated in informal learned-index discussions.  The first theorem is the general residual-entropy law.  The second theorem is the radius-based specialization obtained only under rank-spread.  Throughout this section, $B$ is the local predictor-atom budget from the boxed accounting convention; repair-program tables, directory records, and the sorted array are outside $B$.

\begin{center}
\fbox{\begin{minipage}{0.91\linewidth}
\textbf{Finite-transcript and repair-program accounting.}
The theorem uses finite word-RAM transcripts: predictor coefficients, rounded outputs, certificate endpoints, and local side information are finite words.  The upper-bound direction is a coding construction and may use a transcript-indexed repair program.  That repair program is outside $B$ by convention and must be reported as separate repair-program memory in an implementation.  If it is charged to the same memory resource as local atoms, the theorem becomes a two-resource statement with an explicit repair-program space term.

No compactness claim is made for this repair-program table.  If $Y$ has finite range, let $K_y$ be the number of exact answers attainable by queries in the required domain under transcript $y$.  Only the positive-mass subset of these answers contributes to $H(A\mid Y=y)$, but a materialized repair tree correct on the whole required domain must still represent all $K_y$ attainable answers.  It can be stored with $O(K_y)$ tree nodes for transcript $y$.  Thus the total repair-program description can be as large as $\Theta(\sum_y K_y)$ nodes, and in the worst case as large as the number of finite transcript values times the certified-window support.  This is why the theorem is a local atom-budget time accounting statement within this architecture, not a systems space bound.

The accompanying Lean and Coq files are modest sanity checks for this finite interface.  They machine-check transcript/window/accounting invariants used below: residual offsets are bounded by and recover the certified window position, singleton windows have zero residual ambiguity, repair-program answers remain inside the transcript window, and charging directory or repair-program bytes outside $B$ does not change the local atom count.
\end{minipage}}
\end{center}

\begin{theorem}[Residual-entropy accounting law]
\label{thm:residual-law}
Under the accounting convention above, for every ordered set $S$, workload $\mu$, atomized model family $\M$, and atom budget $B$, every routed piecewise learned index $D$ using at most $B$ local model atoms satisfies
\[
  \E_\mu[T_D]=\Omega(\RGap_\M(S,\mu,B)).
\]
Conversely, for every $\eta>0$ there is a routed piecewise learned index using at most $B$ atoms and satisfying
\[
  \E_\mu[T_D]\le O(\RGap_\M(S,\mu,B)+1)+\eta.
\]
The constants are independent of $S$, $\mu$, and $B$ under the finite-encoding convention above.  The atom budget counts only local predictor records.  Directory topology, directory keys, the sorted array, and any stored transcript-to-repair-program tables are outside $B$ and are accounted for by routing/repair comparisons or by an explicit additional memory term.  No bound in terms of $B$ alone is claimed for the size of a materialized transcript-to-repair-program table.
\end{theorem}

\begin{proof}
Let $D$ induce the ordered leaf partition $\Pi=\{I_1,\ldots,I_m\}$, leaf masses $p_j$, routing depths $d_j$, charged predictors $h_j\in\M_{a_j}$, and certificates $\Delta_j$.  The directory transcript is a binary prefix code for the leaf identity, so \cref{lem:routing-lb} gives
\[
  \E[T_D^{\mathrm{route}}]\ge \sum_j p_jd_j\ge H_\mu(\Pi).
\]
Now condition on the event $q\in I_j$.  Before repair begins the algorithm may know the full pre-repair transcript $Y_{I_j,h_j}(q)$: the leaf, the predictor output, the certified window, and all deterministic information obtained from charged local atoms.  It may not know the outcomes of comparisons between $q$ and stored keys, because those are precisely the repair probes.  By \cref{lem:residual-entropy}, the conditional expected number of repair comparisons is at least
\[
  H(A_{I_j}(q)\mid Y_{I_j,h_j}(q))-O(1)
  =\Rep_\mu(I_j,h_j,\Delta_j)-O(1).
\]
Summing over leaves yields
\[
  \E[T_D]\ge H_\mu(\Pi)+\sum_jp_j\Rep_\mu(I_j,h_j,\Delta_j)-O(1).
\]
Since $\sum_j a_j\le B$, the induced partition and charged predictors are feasible for \cref{eq:rgap-def}.  Taking the infimum over all feasible tuples gives the stated lower bound.

For the upper bound, choose a feasible tuple in \cref{eq:rgap-def} within $\eta$ of the infimum.  Build a nearly optimal alphabetic directory on the selected leaves, with expected depth at most $H_\mu(\Pi)+2$.  The next step is the non-implementation step in the theorem: it is a finite coding construction for the repair phase, not the repair algorithm used by the benchmark.  In the finite word-RAM transcript model, each leaf has a finite set of possible pre-repair transcript values.  For every positive-probability value $y$, consider all $K_y$ exact answers attainable in the required query domain under $y$.  Give the zero-mass attainable answers total auxiliary probability $\varepsilon_y>0$, preserving their order, and rescale the positive conditional distribution by $1-\varepsilon_y$.  A nearly optimal alphabetic comparison tree for this full finite distribution is correct for every attainable answer, and its expected depth on the original conditional distribution is at most $H(A\mid Y=y)+2+o_{\varepsilon_y}(1)$.  For attainable transcript values having zero probability under $\mu$, store any correct finite comparison tree; they contribute memory but no expected-time term.  Because there are finitely many positive-probability transcripts, choose the auxiliary masses so that their averaged overhead is below $\eta$.

This is where the accounting convention matters.  The family of transcript-indexed repair trees is a finite repair program: a dispatch table from transcript codes to comparison-tree descriptions, or equivalently one finite comparison program that branches first on the transcript and then on ordered key comparisons.  Its description is not counted as local predictor atoms and is not hidden in $B$; it is an explicit non-atom resource.  The table may be large.  An explicit description needs $O(K_y)$ tree nodes for a transcript with $K_y$ attainable answers, so a materialized construction can require $\Theta(\sum_y K_y)$ repair-program nodes and admits no general bound in terms of $B$ alone.  If a pure byte-space theorem is desired, this table must be charged as an additional repair-program memory term, just as a systems artifact must report directory and repair-program bytes separately.  With such a charge, the clean $\RGap_\M(S,\mu,B)$ accounting statement becomes a two-resource tradeoff involving both local atoms and repair-program bytes.  Averaging over transcript values and leaves gives repair cost at most $\sum_jp_j\Rep_\mu(I_j,h_j,\Delta_j)+O(1)+\eta$.  For measurable real-output predictors, the implementation first fixes a finite word encoding or entropy-preserving quantization of the transcript; any additional precision used by the repair program is charged outside $B$ by the same convention.  Letting the feasible tuple approach the infimum proves the claim.
\end{proof}

The upper-bound construction is an information-theoretic coding construction.  It does not assert that a production learned index should materialize one bespoke repair tree for every transcript.  Practical systems usually use binary, exponential, or interpolation repair inside the certified window; the benchmark follows this practical route.  The theorem says what repair time is achievable when the finite repair program is accounted for outside the local atom budget, while the experiments report explicit repair implementations and their routing/cache overheads.

\begin{corollary}[Radius surrogate under rank-spread]
\label{cor:rank-spread-gap}
Let $D$ be any rank-spread routed piecewise learned index for $(S,\mu)$ using at most $B$ model atoms from class $\M$.  Then
\[
  \E_\mu[T_D]=\Omega(\Gap_{\M}(S,\mu,B)),
\]
where the constant may depend on the fixed spread constants but not on $n$, $B$, $S$, or $\mu$.  Conversely, there is a routed piecewise learned index with expected query time $O(\Gap_\M(S,\mu,B)+1)$.
\end{corollary}

\begin{proof}
Let $D$ be a rank-spread routed piecewise learned index using local predictors $h_j\in\M_{a_j}$ and certified radii $\Delta_j$ on leaves $I_j$.  Since $\sum_j a_j\le B$, the tuple $(\Pi,\Delta_1,\ldots,\Delta_m)$ is feasible for the optimization defining $\Gap_\M$: by definition of the local profile, $a_j\ge \Comp_{I_j,\M}(\Delta_j)$ up to the fixed atomization convention, and hence $\sum_j\Comp_{I_j,\M}(\Delta_j)\le B$.  By \cref{thm:residual-law}, every such structure has expected cost at least a constant multiple of
\[
  H_\mu(\Pi)+\sum_jp_j\Rep_\mu(I_j,h_j,\Delta_j).
\]
Rank-spread and \cref{lem:spread-surrogate} give
\[
  \Rep_\mu(I_j,h_j,\Delta_j)\ge c\log(1+\Delta_j)
\]
for a constant $c>0$ depending only on the spread constants.  Therefore
\[
  \E_\mu[T_D]\ge
  \Omega\!\left(\sum_jp_j\left(\log\frac{1}{p_j}+\log(1+\Delta_j)\right)\right)
  \ge \Omega(\Gap_\M(S,\mu,B)),
\]
because $\Gap_\M$ is the minimum of the same radius objective over all feasible partitions and radii.

For the upper bound, choose a feasible partition and radii whose objective is within an additive constant of $\Gap_\M(S,\mu,B)$, or within an arbitrary $\eta>0$ if the minimum is not attained.  For each leaf $I_j$, the definition of $\Comp_{I_j,\M}(\Delta_j)$ supplies a certified predictor using at most that many atoms and radius $\Delta_j$.  The total atom count is at most $B$.  Route the leaves using the nearly optimal alphabetic tree from \cref{lem:routing-ub}, which has expected depth at most $H_\mu(\Pi)+2$.  After routing and prediction, binary search inside the certified repair window uses $O(\log(1+\Delta_j))$ comparisons by \cref{lem:repair-ub}.  Averaging over leaves gives expected query time
\[
  O\!\left(H_\mu(\Pi)+\sum_jp_j\log(1+\Delta_j)+1\right)
  =O(\Gap_\M(S,\mu,B)+1).
\]
\end{proof}

\begin{corollary}[Residual accounting within this architecture]
\label{cor:instance-opt}
For routed, atom-budgeted, certified-repair learned indexes, $\RGap_\M(S,\mu,B)$ is the correct instance parameter within this architecture up to constant factors.  The simpler parameter $\Gap_\M(S,\mu,B)$ has the same status only after restricting to rank-spread leaves or after replacing residual entropy by the radius surrogate as an explicit design objective.
\end{corollary}

\begin{proof}
For the routed, atom-budgeted, certified-repair architecture, \cref{thm:residual-law} proves both directions for the residual functional: every admissible structure has expected cost at least a constant multiple of $\RGap_\M(S,\mu,B)$, and there is an admissible structure with expected cost at most a constant multiple of $\RGap_\M(S,\mu,B)+1$ under the stated accounting convention.  Hence $\RGap_\M$ is the controlling instance parameter for this accounting model.  If one replaces the residual entropy inside each certified window by the radius expression $\log(1+\Delta)$, then the replacement is justified exactly by the rank-spread hypothesis via \cref{lem:spread-surrogate}.  Under that additional hypothesis, \cref{cor:rank-spread-gap} gives the same two-sided accounting statement for $\Gap_\M$.  Without rank-spread, only the residual statement remains valid, because a large certified window can have low conditional residual entropy.
\end{proof}

\subsection{Scope of the lower bound}

The lower bound is a theorem about the exact model defined in \cref{sec:model}.  It applies when pieces are selected by ordered branching, local predictors are budgeted by counted atoms, and exactness is obtained by ordered repair inside certified prediction windows.  It does not apply without modification to pure arithmetic dispatch, neural routers, learned hash-routing layers, or any structure that stores exact rank information in an uncharged model or router.

The residual-entropy theorem is the mathematically stable statement.  The radius statement is a corollary that requires rank-spread.  This distinction matters because a large radius and a low-entropy residual distribution can coexist.  A point-query workload concentrated on one residual answer has $\Rep_\mu=0$ even if $\Delta$ is large; a locally uniform workload on the same window has $\Rep_\mu=\Theta(\log(1+\Delta))$.  Thus the residual statement uses $\RGap_\M$ in the full model and uses $\Gap_\M$ only under rank-spread.

The atom budget is also structural.  Without it, exact rank information can be encoded into local predictors and the repair entropy disappears.  With it, the structure must decide where approximation information is worth storing.  The lower bound is therefore not caused by a poor corrective search algorithm; it is caused by the amount of residual exact-answer information that remains after the charged approximation is known.

\section{Power Laws and Allocation}
\label{sec:powerlaw}

The static theorem is intentionally profile-parametric.  This section studies the most useful closed form: local profiles that decay according to a discrete power law until they reach a boundary floor.  The normalization is stated explicitly because unnormalized expressions such as $\Comp(\Delta)\asymp\kappa\Delta^{-\alpha}$ become meaningless once the right-hand side drops below one atom.

\subsection{Closed form under a normalized discrete power law}

The theorem is parameterized rather than distribution-specific.  We now solve the radius-allocation part approximately under \cref{def:powerlaw}.  For a real number $x>0$, write $\logplus x=\max\{0,\log x\}$.

\begin{theorem}[Discrete power-law instance formula]
\label{thm:power-law}
Fix an ordered partition $\Pi=\{I_1,\ldots,I_m\}$ with positive leaf masses $p_1,\ldots,p_m$ and rank diameters $R_1,
\ldots,R_m$.  Suppose $\Pi$ satisfies the normalized profile condition \eqref{eq:powerlaw}.  Let $\Beff=B+m$ be the effective budget after accounting for the one-atom floor in each active leaf.  Then the optimal fixed-partition cost is, up to constants depending only on $a,b,\alpha$,
\[
H_\mu(\Pi)+
\sum_{j=1}^{m}p_j\log\left(1+
\min\left\{R_j,\max\left\{1,\left(\frac{\kappa_j}{\Beff p_j}\right)^{1/\alpha}\right\}\right\}\right).
\]
In the interior regime where $1\ll (\kappa_j/(\Beff p_j))^{1/\alpha}\ll R_j$ for the leaves contributing to the repair term, this simplifies to
\[
  \Theta\left(H_\mu(\Pi)+\frac{1}{\alpha}\sum_{j=1}^{m}p_j\logplus\frac{\kappa_j}{\Beff p_j}\right).
\]
An optimal allocation may be chosen with
\[
  1+\Delta_j^\star=\Theta\left(
  \min\left\{1+R_j,\max\left\{1,\left(\frac{\kappa_j}{\Beff p_j}\right)^{1/\alpha}\right\}\right\}\right)
\]
after integer rounding.
\end{theorem}

\begin{proof}
For the fixed partition the routing term is constant, so it remains to solve the radius allocation problem.  Put $y_j=1+\Delta_j$, with $1\le y_j\le 1+R_j$.  By \cref{def:powerlaw}, a feasible allocation with budget $B$ satisfies
\[
  \sum_j \kappa_j y_j^{-\alpha}=O(B+m)=O(\Beff),
\]
after absorbing the atom floor into $\Beff$; conversely, an allocation satisfying the same inequality with a sufficiently small hidden constant is feasible after changing the budget by only a constant factor.  Thus the discrete problem is sandwiched by the relaxed problem
\[
  \min_{1\le y_j\le 1+R_j}\sum_jp_j\log y_j
  \quad\text{subject to}\quad
  \sum_j\kappa_jy_j^{-\alpha}\le C\Beff
\]
for constants depending only on the profile constants.

For a non-boundary coordinate the KKT condition gives
\[
  \frac{p_j}{y_j}=\lambda\alpha\kappa_jy_j^{-\alpha-1},
  \qquad\text{so}\qquad
  y_j^\alpha=\frac{\lambda\alpha\kappa_j}{p_j}.
\]
The multiplier is chosen so that the relaxed envelope is tight; under the normalized envelope this puts the active scale at $\kappa_j/(\Beff p_j)$ up to constants.  Coordinates whose stationary value is below $1$ are rounded to the atom-floor radius, and coordinates whose stationary value exceeds $1+R_j$ use the full-repair boundary.  Hence an optimal radius satisfies
\[
  y_j=\Theta\!\left(
  \min\left\{1+R_j,\max\left\{1,\left(\frac{\kappa_j}{\Beff p_j}\right)^{1/\alpha}\right\}\right\}\right).
\]
Rounding $y_j-1$ to an integer changes $y_j$ by at most a constant factor away from the floor and changes the objective by only $O(p_j)$ at the floor.  Summing over leaves changes the objective by an additive constant.  Substituting the displayed scale into $\sum_jp_j\log y_j$ gives the capped expression.  When no cap is active on the contributing leaves, $\log y_j=\Theta(\alpha^{-1}\logplus(\kappa_j/(\Beff p_j)))$, which gives the interior formula.
\end{proof}

\begin{remark}[Interpretation]
The term $H_\mu(\Pi)$ is the routing entropy and would already be present in a merely distribution-sensitive decision tree.  The second term measures misalignment between local hardness and query mass.  If an interval is easy to approximate and frequently queried, it receives a tiny repair radius.  If it is hard and rarely queried, it can be assigned a larger radius without hurting the expected cost much.  The expensive case is when many difficult intervals are queried often enough that the atom budget cannot make their repair windows small.
\end{remark}

\begin{corollary}[Uniform hardness recovers entropy]
\label{cor:entropy}
Suppose $\kappa_j=\Theta(1)$ for all leaves and $B=\Theta(m)$ after normalizing the profile so that one atom per leaf gives constant error.  Then
\[
  \Gap_{\M}(S,\mu,B)=\Theta(H_\mu(\Pi)+1)
\]
for the fixed partition $\Pi$.
\end{corollary}

\begin{proof}
Substitute $\kappa_j=\Theta(1)$ and $B=\Theta(m)$ into \cref{thm:power-law}.  The positive-log term is nonzero only on leaves with $p_j\lesssim 1/m$, and its weighted average is bounded by a constant under the normalized linear-budget regime.  The remaining term is $H_\mu(\Pi)$.
\end{proof}

\begin{corollary}[A non-uniform lower bound for hard distributions]
\label{cor:hardnonuniform}
Let a partition $\Pi$ have $m$ leaves of equal mass $1/m$, and suppose each leaf has hardness $\kappa_j\ge\kappa>0$ with $\alpha=1$.  Then every routed piecewise learned index with budget $B$ satisfies
\[
  \E[T]=\Omega\left(\log m+\logplus\frac{\kappa m}{B+m}\right).
\]
In particular, if $B+m=o(\kappa m)$, constant expected query time is impossible within this model.
\end{corollary}

\begin{proof}
For equal masses, $p_j=1/m$ for all $j$, so
\[
  H_\mu(\Pi)=\sum_{j=1}^m\frac{1}{m}\log m=\log m.
\]
With $\alpha=1$ and $\kappa_j\ge\kappa$, the positive-log lower-bound term in \cref{thm:power-law} satisfies
\[
\begin{aligned}
  \sum_{j=1}^{m}p_j\logplus\frac{\kappa_j}{(B+m)p_j}
  &\ge \sum_{j=1}^{m}\frac1m
       \logplus\frac{\kappa}{(B+m)/m}\\
  &=\logplus\frac{\kappa m}{B+m}.
\end{aligned}
\]
Adding the routing entropy term gives the displayed lower bound.  If $B+m=o(\kappa m)$, then $\kappa m/(B+m)\to\infty$, so the second term diverges and constant expected query time is impossible in the rank-spread routed model.
\end{proof}

The preceding corollaries cover the standard benchmark regimes without requiring a separate summary table.  Perfect local prediction collapses the repair term and leaves the alphabetic-search quantity $\Theta(H_\mu(\Pi))$.  A single global predictor with radius $\Theta(n)$ recovers the binary-search scale $\Theta(\log n)$.  A power-law profile yields the explicit entropy--deficiency formula of \cref{thm:power-law}.  Uniform hardness with a linear atom budget returns entropy-optimal behavior, whereas equal-mass hard leaves force the additional term $\Omega(\logplus(\kappa m/(B+m)))$.

\begin{corollary}[Capacity-concentration obstruction]
\label{cor:capacity-obstruction}
Fix a power-law partition $\Pi$ and a threshold $\rho>1$.  Let
\[
  A_\rho=\left\{j:\frac{\kappa_j}{\Beff p_j}\ge \rho\right\},
  \qquad P_\rho=\sum_{j\in A_\rho}p_j.
\]
Then the repair deficiency for this partition is
\[
  \Def_\M(\Pi,B)=\Omega\left(\frac{P_\rho}{\alpha}\log \rho\right).
\]
Consequently, any rank-spread routed piecewise learned index whose induced partition contains such a set of leaves has expected query time at least
\[
  \Omega\left(H_\mu(\Pi)+\frac{P_\rho}{\alpha}\log\rho\right).
\]
\end{corollary}

\begin{proof}
For every $j\in A_\rho$,
\[
  \frac{\kappa_j}{\Beff p_j}\ge\rho,
\]
and therefore the positive-log term in \cref{thm:power-law} satisfies
\[
  p_j\frac{1}{\alpha}\logplus\frac{\kappa_j}{\Beff p_j}
  \ge p_j\frac{1}{\alpha}\log\rho.
\]
Summing this inequality over $j\in A_\rho$ yields
\[
  \Def_\M(\Pi,B)
  =\Omega\left(\sum_{j\in A_\rho}p_j\frac{1}{\alpha}\log\rho\right)
  =\Omega\left(\frac{P_\rho}{\alpha}\log\rho\right).
\]
For a rank-spread routed index inducing this partition, \cref{cor:rank-spread-gap} converts the radius deficiency into an expected-query lower bound.  The routing contribution is $H_\mu(\Pi)$, and the deficiency contribution is the bound above, so the total expected cost is at least
\[
  \Omega\left(H_\mu(\Pi)+\frac{P_\rho}{\alpha}\log\rho\right).
\]
\end{proof}

\subsection{Synthetic examples}

\begin{example}[Uniform workload]
Let a fixed partition have $m$ equal-mass leaves, $p_j=1/m$, equal rank diameters above the optimum radius, and equal hardness $\kappa_j=\kappa$.  For $\alpha=1$, \cref{thm:power-law} gives
\[
  \Gap_\M(\Pi,B)=\Theta\left(\log m+\logplus\frac{\kappa m}{B+m}\right).
\]
Thus entropy is the whole story once the budget scale matches the aggregate local hardness.  If the budget is smaller by a polynomial factor, the missing approximation capacity reappears as an additive repair term.
\end{example}

\begin{example}[Zipf workload]
Let $p_j=j^{-s}/Z_m(s)$ and suppose $\kappa_j=\kappa$ and the optimum is in the interior regime.  The repair deficiency is
\[
  \frac{1}{\alpha}\sum_{j=1}^m p_j\logplus\frac{\kappa}{\Beff p_j}.
\]
For $s>1$, most mass lies on a small prefix; the water-filling rule assigns those leaves small radii and leaves the long tail with larger windows.  For $s\le 1$, the tail mass decays slowly, so more leaves remain above the hardness-to-mass threshold and the repair term can persist even when the entropy term is moderate.
\end{example}

\begin{example}[A locally hard interval]
Suppose one leaf has mass $\tau$ and hardness $K\gg 1$, while all other leaves have unit hardness and share the remaining mass.  In the interior regime the hard leaf contributes
\[
  \Omega\left(\frac{\tau}{\alpha}\logplus\frac{K}{\Beff\tau}\right)
\]
by itself.  A single high-curvature or noisy region can therefore dominate the learned-index cost when it receives non-negligible query mass.  If the same hard region is rarely queried, the contribution is suppressed by the multiplier $\tau$ and the optimizer rationally leaves a larger repair window there.
\end{example}

\subsection{Miniature profile calculation and allocation pseudocode}

The profile abstraction becomes concrete once local segment-cover curves have been estimated.  Consider $m=32$ candidate leaves with rank diameter at least $256$, exponent $\alpha=1$, and normalized profiles
\[
  1+\Comp_j(\Delta)\asymp 1+\kappa_j(1+\Delta)^{-1}.
\]
Using the closed form, the predicted fixed-partition cost is $H_\mu(\Pi)$ plus the weighted repair term.  For a uniform workload with $\kappa_j=256$ for every leaf, the entropy part is $5.00$ bits and the predicted repair terms at budgets $B=32,64,128,256$ are approximately
\[
  7.01,\quad 6.43,\quad 5.71,\quad 4.88.
\]
For a Zipf workload with exponent $1.2$ and the same hardness, the entropy part drops to about $3.75$ bits and the corresponding repair terms are
\[
  5.66,\quad 5.18,\quad 4.54,\quad 3.78.
\]
The decrease is not because the rank curve became easier; it is because query mass is concentrated on fewer leaves, so the optimizer spends atoms there first.  For a locally hard workload with one leaf of mass $0.25$ and hardness $4096$, while the other leaves have hardness $64$, the entropy part is about $4.53$ bits and the repair terms at the same budgets are
\[
  6.05,\quad 5.48,\quad 4.77,\quad 3.98.
\]
The hard leaf dominates only because it also carries substantial mass.  If its mass is moved into the tail, its contribution is suppressed by the factor $p_j$.

A discrete implementation does not need to fit a power law.  It can work directly with certified segment-cover values.  The following pseudocode is the operational version of \cref{prop:separable}.
\begin{quote}
\small
\begin{verbatim}
Input: candidate leaves I_j, workload masses p_j,
       certified profile values C_j[Delta], atom budget B
For every leaf j and every refinement Delta_old -> Delta_new:
    gain  = p_j * (log(1 + Delta_old) - log(1 + Delta_new))
    cost  = C_j[Delta_new] - C_j[Delta_old]
    score = gain / cost
Start from trivial full-repair radii.
Take refinements in decreasing score while the atom budget allows.
Rebuild the ordered directory for the selected leaves.
\end{verbatim}
\end{quote}
This greedy form is exact when refinement options form a laminar concave envelope and is otherwise a practical approximation to the Lagrange multiplier rule.  The mathematical certificate remains the dual inequality \eqref{eq:dual-cert}: a proposed shadow price $\lambda$ proves that no allocation whose remaining refinements all have score below $\lambda$ can improve the objective by more than the unspent budget times $\lambda$.

\subsection{System-facing allocation rule}

The shadow-price formula gives a direct design recipe.  Estimate $p_j$ from the workload, compute or sandwich $\Comp_{I_j,\M}(\Delta)$ at a grid of error radii, and choose a multiplier $\lambda$.  For each leaf, select
\[
  \Delta_j(\lambda)\in\argmin_\Delta\{p_j\log(1+\Delta)+\lambda\Comp_{I_j,\M}(\Delta)\}.
\]
Decrease $\lambda$ until the atom budget is exhausted.  Equivalently, sort all candidate local refinements by
\[
  \frac{p_j(\log(1+\Delta_{\mathrm{old}})-\log(1+\Delta_{\mathrm{new}}))}
       {\Comp_{I_j,\M}(\Delta_{\mathrm{new}})-\Comp_{I_j,\M}(\Delta_{\mathrm{old}})}
\]
and take refinements while the ratio exceeds the current shadow price.  In operational terms: allocate more segments to high-mass regions when their profile curve drops quickly; tolerate larger $\varepsilon$-windows in low-mass or intrinsically hard regions; split a region only when the split lowers routing entropy plus repair deficiency more than it increases directory and model cost.

A direct implementation is:
\begin{enumerate}[label=\arabic*.,leftmargin=1.5em]
  \item choose candidate leaves and a grid of radii;
  \item compute certified profile values by segment-cover routines or lower/upper profile sandwiches;
  \item estimate conditional masses from the workload trace;
  \item greedily or via a Lagrange multiplier choose the local refinements with largest weighted entropy decrease per additional atom;
  \item rebuild the ordered directory for the resulting leaves.
\end{enumerate}

\subsection{Partition selection as a coupled optimization}

The definition of $\Gap_{\M}$ optimizes over both a partition and a radius allocation.  The radius allocation is separable after Lagrangian relaxation, but the partitioning problem remains global because adjacent keys compete for the same model atoms.  For many concrete model classes, however, the partition search has useful structure.  If $\Comp_{I,\M}(\Delta)$ can be evaluated for every interval $I$ and radius $\Delta$, then a dynamic program over prefixes computes the optimum for a fixed number of leaves and a discretized budget.  The state records the right endpoint of the current prefix, the number of used leaves, and the used atom budget; the transition chooses the previous cut and the local radius.  This dynamic program is intended as an exact optimization formulation for the stated architecture rather than a systems implementation; its role is to show that the parameter is algorithmically well-defined rather than merely existential.

For piecewise-linear predictors, the local profile can often be approximated by the minimum number of segments needed to cover the rank curve within vertical error $\Delta$.  Classical algorithms for error-bounded polygonal approximation then provide upper bounds \cite{imai1986polygonal} on $\Comp_{I,\M}(\Delta)$, while approximation-theoretic arguments provide lower bounds.  The instance parameter can therefore be estimated by combining segmentation routines with empirical workload masses.  This provides a bridge from the theorem to systems practice: one can diagnose whether a workload is hard because the query mass is spread across many leaves, because the local rank curve is difficult to approximate, or because the two are poorly aligned.

The Lagrangian form yields a design principle.  A structure should refine a leaf only when the decrease in weighted repair entropy exceeds the atom cost.  Equivalently, the marginal value of a new model atom should be approximately equalized across active leaves.  Regions whose query mass is high and whose approximation profile drops quickly should receive more precise models.  Regions whose query mass is low or whose approximation profile drops slowly should tolerate larger repair windows.  This rule is absent from worst-case learned-index analyses, but it is exactly what an optimizer needs when memory is the scarce resource.

A final structural point is worth isolating.  The best partition for routing entropy need not be the best partition for approximation.  Entropy favors high-probability leaves that can be reached quickly; approximation favors cuts at places where the rank curve changes behavior.  The parameter $\Gap_{\M}$ forces these objectives into the same optimization.  This is why the theorem is more informative than combining an optimal alphabetic tree with an independently optimized model error bound: the two decisions must be made jointly.

\section{Diagnostic Benchmark Results}
\label{sec:experiments}

These experiments have two distinct roles.  First, an exact finite-parameter experiment computes $\RGap_\M$ and $\Gap_\M$, rather than a proxy, on real-data-induced finite instances for an explicitly declared affine atom family and finite transcript encoding.  Second, a full-scale systems benchmark compares production baselines and diagnostic shadow variants on the large datasets.  That second experiment does not claim to compute the optimum for the PGM-index or RadixSpline implementation families; it exposes the systems overheads that separate the mathematical repair objective from end-to-end latency.  Neither layer proves the theorem, which is a formal statement.  Raw datasets are downloaded by script and verified by checksum; they are not included in the artifact archive.

The empirical claim is correspondingly precise.  The exact experiment tests that the theorem parameter is computationally measurable once $\M$, the transcript encoding, $B$, and the finite workload are declared.  The full-scale benchmark tests whether its components and failure modes remain observable in practical implementations: rank-spread can fail on skewed or adversarial workloads, and a prototype that reduces repair comparisons can still lose end-to-end latency because directory dispatch, cache behavior, and branch behavior dominate.  The shadow prototype is a diagnostic allocation instrument, not a proposed replacement for PGM-index or RadixSpline.

\begin{table}[H]
\centering
\small
\caption{Empirical hypotheses and their relation to the theorem.}
\label{tab:benchmark-bridge}
\begin{tabular}{llll}
\hline
ID & Question & Quantity & Evidence \\
\hline
H1 & Exact measurability & $\RGap_\M$, $\Gap_\M$ & finite DP values \\
H2 & Spread dependence & $\RGap_\M$ vs. $\Gap_\M$ & trace diagnostics \\
H3 & Non-atom cost & outside $B$ & bytes, route comps. \\
H4 & Latency separation & repair vs. time & timing ablations \\
\hline
\end{tabular}
\end{table}
\FloatBarrier

\subsection{Exact finite theorem-parameter experiment}

To compute the parameter itself, we define a finite experiment model $\M^{\mathrm{aff}}_Q$.  For each of the four real sorted datasets, the evaluator deterministically extracts $1024$ ordered keys spanning the file; these keys and the indicated finite successful-query distribution constitute the finite instance.  A leaf may use either no predictor at zero atoms or one affine chord predictor at one atom.  Its predicted rank is rounded to a rank quantum of $Q=16$ before the resulting finite word code enters the transcript, and its certified radius is checked against every key in the finite instance.  Thus the transcript alphabet and all repair supports are finite and explicit, rather than estimated by a histogram.

For each workload in $\{\texttt{uniform\_hits},\texttt{hotspot\_hits},\texttt{zipf\_hits}\}$, 50{,}000 deterministic finite-support queries define $\mu$, while correctness of this finite experiment is required for successful lookup of all $1024$ selected keys.  A dynamic program enumerates all ordered intervals and budgets $B\in\{0,4,8,16,32,64,128\}$ and minimizes separately the residual-entropy and radius objectives in the definitions of $\RGap_{\M^{\mathrm{aff}}_Q}$ and $\Gap_{\M^{\mathrm{aff}}_Q}$.  Consequently, the reported values are exact for this declared finite successful-key model and workload; they are not an estimator of a larger model family or of unsuccessful predecessor queries.  The program-node column reports $\sum_y K_y$ for the selected residual construction, where $K_y$ counts answers attainable by the required finite query universe under transcript $y$, including attainable zero-mass answers.  This is the materialized transcript-indexed repair-program resource that the theorem deliberately reports outside $B$.

\begin{table}[H]
\centering
\small
\caption{Exact finite-instance parameters in bits per query for $\M^{\mathrm{aff}}_{16}$.  The selected column budget is $B=32$; program nodes are outside $B$.}
\label{tab:exact-parameter}
\resizebox{\linewidth}{!}{%
\begin{tabular}{llrrrrr}
Dataset & Workload & $\RGap_\M(0)$ & $\RGap_\M(32)$ & $\Gap_\M(32)$ & Atoms & Prog. nodes \\
\hline
Books & hits & 9.98 & 4.32 & 5.92 & 2 & 1024 \\
Books & hotspot & 7.66 & 3.93 & 4.29 & 3 & 1024 \\
Books & zipf & 9.57 & 4.38 & 5.25 & 2 & 1024 \\
\hline
FB & hits & 9.98 & 3.98 & 3.47 & 1 & 1024 \\
FB & hotspot & 7.66 & 3.88 & 3.37 & 1 & 1024 \\
FB & zipf & 9.57 & 3.88 & 3.46 & 1 & 1024 \\
\hline
Wiki & hits & 9.98 & 4.37 & 5.49 & 2 & 1024 \\
Wiki & hotspot & 7.66 & 4.20 & 3.94 & 3 & 1024 \\
Wiki & zipf & 9.57 & 4.04 & 6.09 & 1 & 1024 \\
\hline
OSM & hits & 9.98 & 5.91 & 7.49 & 3 & 1024 \\
OSM & hotspot & 7.66 & 4.66 & 5.39 & 5 & 1024 \\
OSM & zipf & 9.57 & 5.26 & 7.29 & 2 & 1024 \\
\hline

\end{tabular}
}
\end{table}
\FloatBarrier

This table gives a direct empirical link to the theorem within the declared finite architecture: adding local affine atoms reduces exact $\RGap_\M$ on every displayed dataset/workload pair.  It also makes the accounting limitation measurable rather than rhetorical: the chosen upper-bound repair construction retains all 1024 required successful-key answer nodes even at $B=32$.  Thus the improvement is an atom-budget statement, not a byte-space result.  The gaps between $\RGap_\M(32)$ and $\Gap_\M(32)$ further demonstrate why a radius-only interpretation requires the rank-spread condition.

\subsection{Full-scale systems protocol}

The default benchmark suite uses four SOSD/Zenodo sorted integer datasets \cite{marcus2021benchmarking,sortedints2025zenodo}: books, Facebook identifiers, wiki timestamps, and OSM cell identifiers.  In artifact notation these are
\begin{quote}
\small
\texttt{books\_200M\_uint32}, \texttt{fb\_200M\_uint64},\\
\texttt{wiki\_ts\_200M\_uint64}, and \texttt{osm\_cellids\_800M\_uint64}.
\end{quote}
The benchmark harness compares exact lookup paths: \texttt{std::lower\_bound}, PGM-index \cite{ferragina2020pgm}, RadixSpline \cite{kipf2020radixspline}, an ordered shadow-profile prototype, and a radix-routed shadow-profile variant.  The shadow variants are diagnostic prototypes for the allocation rule, not proposed replacements for PGM-index or RadixSpline.  They greedily split high-mass, high-error regions according to the weighted repair decrease per additional segment, so their role is to expose the size/repair/routing tradeoff in the accounting model and to show where that tradeoff stops predicting latency.  The PGM and RadixSpline error grid is $\varepsilon\in\{32,128,512\}$; RadixSpline uses 18 radix bits.  The shadow-profile atom grid is $\{256,1024\}$ with 8192 profile samples per candidate segment.  Unless otherwise stated, each configuration uses 200,000 queries, three independent query seeds, a 1000-query warm-up, and a 20,000-query latency sample.  The workload suite contains successful uniform lookups, unsuccessful midpoint probes, a 50--50 mixed hit/miss workload, Zipf-like hits, hotspot hits, and adversarial probes into large gaps.  All methods return the exact lower-bound position, and the CSV output includes a checksum of returned ranks.

The reported local run used an Apple M4 Pro with 48 GiB RAM, macOS 26.3, Apple clang 17.0.0, and release builds.  Because this is a single-machine artifact rather than a cross-platform performance study, absolute latency values should be read as local measurements.  The intended performance comparison is within this controlled run: repair work, routing work, directory bytes, and latency are recorded together so that a reader can see when the mathematical repair term stops predicting end-to-end time.  The paper reports confidence intervals over repeated query streams and keeps raw CSV output in the artifact for reanalysis.

\subsection{Full-scale systems results}

Each full-scale run records build time, total index bytes, directory bytes outside the atom budget, repair-program bytes, model atoms or segment records, mean lookup latency, sampled 95th-percentile latency, average repair comparisons, routing comparisons, fallback count, coarsened residual entropy ratio, residual support, and checksum.  Unlike \cref{tab:exact-parameter}, these trace diagnostics are not exact values of $\RGap_\M$ for the full implementation families.  \Cref{tab:real-results} reports the fastest successful-lookup configuration within each family.  The full grid is retained in \texttt{results/benchmark\_results.csv}.

\begin{table}[H]
\centering
\small
\caption{Real-data diagnostic successful-lookup results.  Size is total index size excluding the sorted array; Dir is the directory part outside the local atom budget.  Mean, build, and p95 cells report mean $\pm$ 95\% confidence interval over repeated query streams when repeats are available.}
\label{tab:real-results}
\resizebox{\linewidth}{!}{%
\begin{tabular}{llrrrrrr}
Dataset & Index & Size MB & Dir MB & Build ms & Mean ns & p95 ns & Repair \\
\hline
Books & Binary & 0.00 & 0.00 & 0.0 & 497.9{\scriptsize$\pm$7.8} & 1430.3{\scriptsize$\pm$98.3} & 27.66 \\
Books & Shadow-O(1024) & 0.07 & 0.01 & 251.7{\scriptsize$\pm$16.7} & 743.5{\scriptsize$\pm$34.9} & 791.7{\scriptsize$\pm$47.5} & 10.60 \\
Books & PGM(32) & 4.12 & 0.00 & 1730.6{\scriptsize$\pm$6.8} & 177.4{\scriptsize$\pm$14.9} & 402.7{\scriptsize$\pm$54.2} & 6.08 \\
Books & RS(32) & 11.72 & 0.00 & 1069.2{\scriptsize$\pm$4.4} & 181.7{\scriptsize$\pm$6.3} & 402.7{\scriptsize$\pm$54.2} & 6.09 \\
\hline
FB & Binary & 0.00 & 0.00 & 0.0 & 458.9{\scriptsize$\pm$7.3} & 1347.3{\scriptsize$\pm$151.5} & 27.66 \\
FB & Shadow-R(1024) & 2.16 & 2.11 & 214.1{\scriptsize$\pm$10.5} & 1336.2{\scriptsize$\pm$45.9} & 1667.0{\scriptsize$\pm$124.7} & 14.66 \\
FB & PGM(32) & 17.31 & 0.00 & 2455.5{\scriptsize$\pm$15.3} & 243.3{\scriptsize$\pm$15.5} & 513.7{\scriptsize$\pm$54.6} & 6.07 \\
FB & RS(128) & 9.11 & 0.00 & 1098.6{\scriptsize$\pm$1.8} & 416.3{\scriptsize$\pm$22.9} & 680.7{\scriptsize$\pm$26.8} & 8.02 \\
\hline
Wiki & Binary & 0.00 & 0.00 & 0.0 & 463.4{\scriptsize$\pm$14.8} & 1250.0{\scriptsize$\pm$46.4} & 27.66 \\
Wiki & Shadow-O(1024) & 0.07 & 0.01 & 211.8{\scriptsize$\pm$1.3} & 1281.2{\scriptsize$\pm$9.8} & 1292.0 & 14.73 \\
Wiki & PGM(32) & 1.38 & 0.00 & 1429.4{\scriptsize$\pm$1.4} & 167.0{\scriptsize$\pm$7.7} & 347.3{\scriptsize$\pm$72.1} & 6.08 \\
Wiki & RS(32) & 4.16 & 0.00 & 644.5{\scriptsize$\pm$4.5} & 174.3{\scriptsize$\pm$13.3} & 375.3{\scriptsize$\pm$47.0} & 6.08 \\
\hline
OSM & Binary & 0.00 & 0.00 & 0.0 & 745.5{\scriptsize$\pm$118.7} & 1833.3{\scriptsize$\pm$188.4} & 29.66 \\
OSM & Shadow-R(1024) & 2.16 & 2.11 & 386.0{\scriptsize$\pm$22.2} & 2151.0{\scriptsize$\pm$77.9} & 2583.0 & 19.07 \\
OSM & PGM(32) & 47.47 & 0.00 & 9558.7{\scriptsize$\pm$98.1} & 305.3{\scriptsize$\pm$18.7} & 569.3{\scriptsize$\pm$54.6} & 6.07 \\
OSM & RS(512) & 5.82 & 0.00 & 4371.6{\scriptsize$\pm$34.0} & 402.7{\scriptsize$\pm$33.9} & 819.3{\scriptsize$\pm$54.6} & 10.01 \\
\hline

\end{tabular}
}
\end{table}
\FloatBarrier

The results should be read as a diagnostic comparison, not as a speed claim for the shadow prototype.  PGM and RadixSpline reduce the repair term to about six to ten comparisons and obtain the best latency on most real-data configurations.  The shadow-profile variants also reduce repair comparisons relative to binary search, often with far smaller model state, but ordered routing and radix dispatch add costs outside the theorem's repair term.  The latency lesson is short: repair comparisons are necessary but not sufficient.  Once repair is small, directory dispatch, cache behavior, and branch behavior can dominate.  The detailed stress-workload, rank-spread, routing-overhead, and correlation tables are moved to \cref{app:benchmark-diagnostics} so that the main text keeps the experimental message rather than a performance scoreboard.

The synthetic allocation examples in \cref{sec:powerlaw} explain why the shadow rule spends atoms where query mass and local hardness align.  The real-data results add the systems qualification: a lower repair term must be priced alongside non-atom directory bytes and hardware effects before it becomes an end-to-end latency prediction.

\subsection{Reproducibility package}

The artifact contains a C++17 full-scale benchmark harness, a C++17 exact finite-parameter evaluator, Python scripts for downloading datasets and aggregating CSV outputs into LaTeX tables, and formal sanity checks in both Lean 4 and Coq.  These proof files have a deliberately limited role: they check the finite transcript, repair window, and accounting interface used by the theorem.  In particular, they prove that a certified singleton repair window has zero residual ambiguity, that residual offsets are bounded by the certified window width, that a repair program's answer lies inside the transcript window, and that materializing arbitrary repair-program nodes changes the separately reported repair resource without changing the local atom count.  They do not mechanize the entropy theorem, the alphabetic-tree coding argument, the dynamic program, or the benchmark algorithms.  A standard run executes:

\begin{table}[H]
\centering
\small
\caption{Artifact reproducibility metadata.  The manuscript archive is self-contained except for raw datasets, which are downloaded and verified by checksum.}
\label{tab:artifact-metadata}
\resizebox{\linewidth}{!}{%
\begin{tabular}{ll}
\hline
Item & Value \\
\hline
Artifact archive & \texttt{learned\_index\_artifact.zip} \\
Archive provenance & standalone submission archive; no enclosing Git commit in this workspace \\
PGM-index snapshot & \texttt{c6fcf3d} from \texttt{github.com/gvinciguerra/PGM-index} \\
RadixSpline snapshot & \texttt{ab96aa5} from \texttt{github.com/learnedsystems/RadixSpline} \\
Compiler/toolchain & Apple clang 17.0.0, C++17, CMake release build \\
C++ flags & \texttt{-O3 -Wall -Wextra -Wpedantic} \\
Benchmark grid & 200{,}000 queries, 3 query seeds, 1000 warm-up queries, 20{,}000 latency samples \\
Exact-parameter grid & 1024 keys, 50{,}000 queries, $Q=16$, $B\in\{0,4,8,16,32,64,128\}$ \\
Formal checks & \texttt{lean formal/ResidualEntropySanity.lean}; \texttt{coqc formal/ResidualEntropySanity.v} \\
\hline
\end{tabular}
}
\end{table}
\FloatBarrier

\begin{table}[H]
\centering
\small
\caption{Dataset checksums used by the artifact downloader.  Raw datasets are excluded from the archive and reproduced from Zenodo record \href{https://doi.org/10.5281/zenodo.15240501}{10.5281/zenodo.15240501}.}
\label{tab:dataset-checksums}
\resizebox{\linewidth}{!}{%
\begin{tabular}{lrl}
\hline
Dataset & Bytes & MD5 \\
\hline
\texttt{books\_200M\_uint32} & 800000008 & \texttt{55845580be1554d82be1c0dda416005c} \\
\texttt{fb\_200M\_uint64} & 1600000008 & \texttt{679eff3bfbc80572b30f6575b40b6918} \\
\texttt{wiki\_ts\_200M\_uint64} & 1600000008 & \texttt{4f1402b1c476d67f77d2da4955432f7d} \\
\texttt{osm\_cellids\_800M\_uint64} & 6400000008 & \texttt{70670bf41196b9591e07d0128a281b9a} \\
\hline
\end{tabular}
}
\end{table}
\FloatBarrier

\begin{quote}
\small
\begin{verbatim}
scripts/fetch_third_party.sh
cmake -S . -B build -DCMAKE_BUILD_TYPE=Release
cmake --build build --target learned_index_bench -j
cmake --build build --target exact_parameter -j
python3 scripts/download_sosd.py --data-dir data
scripts/run_benchmarks.sh
scripts/run_exact_parameter.sh
lean formal/ResidualEntropySanity.lean
coqc formal/ResidualEntropySanity.v
scripts/check_artifact_paths.sh
\end{verbatim}
\end{quote}
The path-scrub check rejects generated source artifacts containing local absolute paths.  Thus benchmark outputs can be archived without exposing workstation-specific directories.  The smoke-test configuration runs on a synthetic dataset and verifies exactness against binary search before any large real-data run is launched.

\section{Dynamic Boundary Checks}
\label{sec:dynamic}

The static theorem concerns a fixed set and a fixed workload.  Dynamic updates change both the rank curve and the workload-conditioned profile.  The results in this section should therefore be read as consequences and boundary checks, not as the central novelty of the paper.  They explain what the static framework predicts under standard rebuilding, where it loses force under tombstones and drift, and why exact dynamic rank still inherits classical cell-probe barriers.

\subsection{Logarithmic rebuilding}

Fix an integer branching parameter $\beta\ge 2$.  Maintain levels
\[
  L_0,L_1,\ldots,L_h,\qquad h=\lfloor\log_\beta n\rfloor,
\]
where level $L_i$ stores either zero or $\Theta(\beta^i)$ live keys in sorted order together with a static routed piecewise learned index built for that level.  Insertions behave like carries in a base-$\beta$ counter: when $\beta$ structures of size $\beta^i$ become full, they are merged into one structure of size $\beta^{i+1}$.  Deletions can be handled by symmetric rebuilding or by tombstones followed by periodic cleanup.  This is the cleanest dynamization to analyze because each key participates in only $O(\log_\beta n)$ rebuild levels.

\begin{theorem}[Standard logarithmic-rebuilding consequence]
\label{thm:dynamic-upper}
Assume that a static structure on $m$ keys can be rebuilt in $O(m)$ time and uses $O(m)$ atoms.  The $\beta$-ary logarithmic method yields a fully dynamic routed piecewise learned index with:
\begin{enumerate}[label=(\alph*),leftmargin=1.5em]
  \item total space $O(n)$;
  \item amortized insertion and deletion time $O(\beta\log_\beta n)$, ignoring the cost of external memory allocation;
  \item exact predecessor and rank queries in expected time
  \[
    O\left(\sum_{i\in A}\bigl(\Gap_{\M}(S_i,\mu_i,\Theta(|S_i|))+1\bigr)\right),
  \]
  where $A$ is the set of active levels, $S_i$ is the live set stored at level $i$, and $\mu_i$ is the induced query distribution for level $i$.
\end{enumerate}
In particular, if every active level has static gap at most $G$, then queries take $O((G+1)\log_\beta n)$ expected time.
\end{theorem}

\begin{proof}
Each live key is stored in exactly one active level, and the static structure stored at level $i$ uses $O(|S_i|)$ atoms by assumption.  Since the active levels are disjoint as sets of keys, the total number of stored keys over all levels is $n$ and the total space is $O(n)$.

For updates, consider the carry operation at level $i$.  A rebuild from $\beta$ full structures of size $\Theta(\beta^i)$ creates one structure of size $\Theta(\beta^{i+1})$ and costs $O(\beta^{i+1})$ by the assumed linear rebuild bound.  Such a carry can occur only after $\Theta(\beta^i)$ lower-level insertions or deletions have accumulated enough items to fill the $\beta$ children that trigger the merge.  Charging the rebuild cost to those triggering updates gives an amortized contribution $O(\beta)$ for level $i$.  There are $O(\log_\beta n)$ possible levels, so the total amortized update cost is $O(\beta\log_\beta n)$.

For queries, the structure asks every active level for its local predecessor candidate, or for its local rank contribution in the rank-query case.  The global predecessor is the maximum of the per-level candidates below $q$, and the global rank is the sum of the per-level ranks.  Let $A$ be the set of active levels.  On level $i$, the static expected query cost under the induced conditional workload $\mu_i$ is bounded by
\[
  O\bigl(\Gap_{\M}(S_i,\mu_i,\Theta(|S_i|))+1\bigr)
\]
by the static rank-spread guarantee applied to that level's structure and atom budget.  Summing the costs over active levels and adding only constant work per level to combine candidates gives the displayed query bound.  If every active level has static gap at most $G$, the sum over $O(\log_\beta n)$ active levels is $O((G+1)\log_\beta n)$.
\end{proof}

\begin{remark}[Choice of $\beta$]
With constant $\beta$, updates cost $O(\log n)$ amortized time and queries scan $O(\log n)$ levels.  Larger $\beta$ reduces the number of active levels but increases rebuild cost.  This is the standard logarithmic-method tradeoff, now with the static cost of each level expressed by $\Gap_\M$.
\end{remark}

\subsection{A level-scanning lower bound}

The previous theorem exposes an additive level-scanning overhead.  The next theorem shows that this is not merely a loose analysis for this family.

\begin{theorem}[Lower bound within the levelled family]
\label{thm:level-lb}
Consider any exact dynamic predecessor or rank structure obtained by a $\beta$-ary levelled decomposition with $h+1$ active levels.  There exists a set of live keys and a query $q$ for which the structure must inspect all $h+1$ active levels.  Consequently, the additive query overhead in this family is $\Omega(h)=\Omega(\log_\beta n)$.
\end{theorem}

\begin{proof}
Let the active levels be $L_0,\ldots,L_h$.  Choose live keys
\[
  y_0<y_1<\cdots<y_h<q
\]
with $y_i\in L_i$ for every $i$.  In this instance, the predecessor of $q$ is $y_h$.  Suppose an exact query algorithm does not inspect some active level $L_t$ on this query.  Construct a second instance that is identical on every inspected level and differs only inside $L_t$: replace the key $y_t$ by a key $z$ satisfying
\[
  y_h<z<q.
\]
The query transcript is unchanged, because the algorithm never reads or otherwise inspects $L_t$ and all inspected state is identical.  However, the correct predecessor in the modified instance is now $z$, not $y_h$.  The algorithm must therefore return the same answer on two indistinguishable executions with different correct answers, contradicting exactness.  Hence every exact algorithm in this levelled family must inspect all $h+1$ active levels on some instance.  Since $h=\Theta(\log_\beta n)$ in the worst case, the additive level-scanning overhead is $\Omega(\log_\beta n)$.
\end{proof}

This theorem is intentionally modest.  It does not claim that every dynamic learned index must pay $\Omega(\log n)$ additive query overhead.  It says that the basic Overmars-style decomposition pays this cost unless it is augmented by additional global coordination.  Recent geometric worst-case constructions show that more sophisticated dynamic learned indexes can avoid pure level scanning at the price of heavier machinery \cite{gaede2025dynamic}.

\subsection{Tombstones and range-reporting degeneration}

A common practical deletion strategy is to insert tombstones and reconcile them during later rebuilds.  This can be efficient for point updates, but it can be bad for exact range reporting.

\begin{proposition}[Tombstone degeneration]
\label{prop:tombstones}
Any levelled dynamic learned index that handles deletions by inserting tombstones can be forced into a state in which an exact range query reporting $k$ live keys must examine $\Omega(N+k)$ items after $N$ update operations, unless the structure performs additional cleanup whose cost is charged elsewhere.
\end{proposition}

\begin{proof}
Fix a query range $Q$.  Perform $N/2$ insertions of distinct keys inside $Q$, followed by $N/2$ deletions of those same keys, and suppose the deletion policy records the removals by tombstones rather than physically cleaning the affected range.  Then the representation contains $\Theta(N)$ stale markers associated with keys in $Q$.  Now leave $k$ live keys inside the same range.  Consider an exact range-reporting query for $Q$.  To certify that its output contains exactly the live keys and no deleted keys, the query procedure must distinguish live records from stale tombstone records encountered in the range.  If it skips an uninspected stale/live marker location, an adversary can change that marker from stale to live, or from live to stale, without changing anything the algorithm has read; the correct output changes but the transcript does not.  Therefore exactness forces the procedure to examine the accumulated stale markers plus the $k$ live output records, which is $\Omega(N+k)$ work unless a prior cleanup step has already paid to remove or summarize those tombstones.
\end{proof}

This pathology matches the practical observation that dynamic learned indexes may behave well on insertion-heavy workloads but degrade under deletion-heavy or adversarial update sequences \cite{wongkham2022ready,yang2024aca,luo2025robust}.

\subsection{Dynamic distribution drift}

The dynamic setting introduces a problem that is absent from the static theorem: the instance itself changes over time.  Insertions and deletions alter the ordered set $S$, but they may also change the workload distribution and the local hardness profile.  A model that is near-optimal for yesterday's rank curve may allocate atoms to the wrong regions after a burst of updates.  In our notation, both $p_j$ and $\kappa_j$ can drift.  When drift preserves their alignment, learned indexing can remain fast; when drift makes high-mass regions locally hard, the repair term grows even if the routing directory remains balanced.

This observation gives a theoretical explanation for why retraining policies matter.  Rebuilding too often pays update cost; rebuilding too rarely lets the structure operate with stale radii and stale partitions.  A principled policy should monitor the increase in the repair-entropy potential and rebuild when the saved query cost amortizes the rebuild cost.  The resulting strategy resembles classical global rebuilding, but the trigger is not only size imbalance.  It is the loss of alignment between workload mass and approximation hardness.

The same perspective applies to deletion handling.  Tombstones preserve update speed by postponing cleanup, but they change the effective repair problem for range queries.  The model may still predict where the range begins, yet exact reporting must separate live keys from stale markers.  In other words, tombstones add a second repair term that is not captured by rank-prediction error alone.  Any fully dynamic theory that aims to be predictive for systems will need to price this stale-state repair explicitly.

\subsection{Cell-probe barrier for exact dynamic rank}

We now leave the levelled family and state a generic exact lower bound.  No learned prediction mechanism can bypass the information-transfer barrier for exact dynamic rank.

\begin{theorem}[Dynamic rank barrier, as a reduction]
\label{thm:rank-barrier}
Any exact dynamic learned data structure that supports insert, delete, and rank on an ordered set of size $n$ in the cell-probe model with $w$-bit cells inherits the classical dynamic subset-rank lower bound.  In particular, in the standard word-RAM regime $w=\Theta(\log n)$, either update time or query time is $\Omega(\log n/\log\log n)$ for the corresponding dynamic prefix-counting primitive.
\end{theorem}

\begin{proof}
Reduce from dynamic subset rank on universe $[n]$.  Maintain a dynamic ordered set
\[
  S_A=\{i\in[n]: A_i=1\}
\]
corresponding to the current bitvector $A\in\{0,1\}^n$.  A bit flip from $0$ to $1$ is an insertion of key $i$, and a bit flip from $1$ to $0$ is a deletion of key $i$.  A subset-rank query asks for
\[
  \rank_A(t)=|\{i\le t:A_i=1\}|.
\]
This is exactly the rank of $t$ in the ordered set $S_A$, up to the standard convention about whether rank counts keys $\le t$ or keys $<t$; if the learned index uses the latter convention, query $t+1$ or add the membership bit of $t$ to translate between the two.  Thus any exact dynamic learned index supporting insert, delete, and rank on an ordered set gives a dynamic subset-rank data structure with the same cell-probe update and query costs, up to constant overhead.

The cell-probe model charges only memory probes and gives computation for free, so the prediction mechanism cannot invalidate the reduction: arbitrary learned computation between probes is already allowed.  Classical dynamic prefix-counting and subset-rank lower bounds therefore apply to the learned-index structure.  In the standard word-RAM regime $w=\Theta(\log n)$, these bounds imply that one of update time or query time must be $\Omega(\log n/\log\log n)$.
\end{proof}

\begin{remark}
The dynamic theorem draws a sharp line between static and dynamic learned indexing.  In the static case, learning can realize a clean residual-entropy advantage by reallocating modeling power to high-mass, low-hardness regions.  In the dynamic case, exact rank still requires maintaining changing order information under adversarial updates.  Prediction can reduce local search, but it cannot eliminate the dynamic information that exactness requires.
\end{remark}

\subsection{Limitations and open problems}
\label{sec:discussion}

The theory suggests a different way to evaluate learned data structures.  The right question is not whether learning beats $O(\log n)$ in the abstract.  The right question is whether the routed, atom-budgeted architecture leaves enough residual exact-answer entropy after the charged predictors have been evaluated.  The parameter $\RGap_{\M}(S,\mu,B)$ is the general answer for that architecture; $\Gap_{\M}(S,\mu,B)$ is the rank-spread radius surrogate.

The main limitation is the routing interface itself.  The theorem covers ordered routing into contiguous intervals followed by counted local prediction and certified repair.  It does not yet give an accounting theorem for unconstrained recursive model indexes, learned hash routers, neural multi-stage dispatch, or arithmetic dispatch schemes that bypass ordered comparisons.  Those designs may still be analyzable by a residual-entropy method, but only after their dispatch records, payloads, and exact-position information are assigned explicit information or byte costs.  This is a limitation of the current theorem, not merely future polish.

The rank-spread specialization is also a real restriction.  The residual theorem remains valid without rank-spread, but the simpler $\log(1+\Delta)$ radius surrogate can fail on hot-key, endpoint, and adversarial-gap workloads.  The finite evaluator therefore reports exact $\RGap_\M$ beside exact $\Gap_\M$ for its declared model, while the full-scale implementation traces report entropy ratios and residual supports rather than assuming that every large repair window implies a large lower bound.

This perspective clarifies several empirical observations.  When systems papers report that the internal search phase of a PGM-like index is memory-bound or highly sensitive to the corrective strategy, they are observing the systems-level signature of the repair term \cite{marcus2021benchmarking,maltry2022critical,liu2025pgmpp}.  When disk-based work adjusts error bounds around page alignment, it is changing the effective cost of a repair window in an I/O model \cite{lan2023disk,zhang2024disk}.  When robustness studies report performance collapse under skewed updates and retraining instability, they are observing that updates change both local hardness and routing masses, often in a misaligned way \cite{yang2024aca,luo2025robust}.

The framework also clarifies why prior theoretical results are complementary rather than contradictory.  Distribution-dependent upper bounds show regimes where the optimum is small \cite{zeighami2023distribution,croquevielle2025constant}.  Approximation-theoretic lower bounds show that certain model families cannot drive $\Comp_{I,\M}(\Delta)$ down quickly enough on uniformly hard regions \cite{croquevielle2026lower}.  Our contribution is to put these two phenomena into a single instance-level objective.

\paragraph{Open problems.}
Several directions seem especially likely to generate follow-up work.

\begin{itemize}[leftmargin=1.5em]
  \item \textbf{Beyond routed piecewise models.}  Extending the same accounting discipline to arithmetic-dispatch structures, learned hash-routing layers, and neural multi-stage routers remains the most important conceptual next step.
  \item \textbf{Measuring $\alpha$ and $\kappa$ for concrete model classes.}  The closed form in \cref{thm:power-law} becomes maximally informative once one can identify local hardness numbers $\kappa_I$ and approximation exponents $\alpha$ for concrete classes of rank curves and concrete families of model atoms.  Establishing such measurements for piecewise-linear predictors on smooth, heavy-tailed, or self-similar inputs would sharpen the theorem's predictive power within this architecture.
  \item \textbf{Dynamic residual guarantees.}  There remains a gap between generic dynamic exact lower bounds, logarithmic-method dynamizations, and geometric worst-case learned indexes.  A natural goal is a fully dynamic learned rank structure with provably near-optimal query distribution sensitivity and worst-case updates near the classical lower-bound frontier, while remaining memory-local enough to be practical.
  \item \textbf{External-memory and cache-aware variants.}  In storage systems the repair cost is not $\log(1+\Delta)$ comparisons but a function of page transfers, cache lines, and prefetching.  A refined $\Gap$ parameter for external memory should replace repair entropy by an architecture-aware correction cost.
\end{itemize}

\paragraph{Conclusion.}
We proved a residual-entropy accounting law for exact predecessor and rank search in atom-budgeted routed piecewise learned indexes.  The theorem's accounting is explicit: ordered routing, counted local model information, certified repair windows, and a separate directory/repair-program memory convention.  Under that accounting, the controlling parameter is routing entropy plus residual answer entropy.  The cleaner $\Gap_\M$ expression is a useful and tight radius surrogate under rank-spread.  The exact finite-instance measurements show that $\RGap_\M$ is computable on real-data-induced workloads once the architecture and transcript are fixed, while also exposing the large repair-program resource outside $B$; the full-scale diagnostic benchmarks show where routing and hardware overheads dominate the repair objective.

\paragraph{Data and artifact availability.}
The benchmark artifact accompanying this manuscript contains the C++17 benchmark harness and exact finite-parameter evaluator, Python aggregation scripts, Lean and Coq sanity checks, generated CSV outputs, and table files.  The public sorted integer datasets are downloaded from Zenodo record \href{https://doi.org/10.5281/zenodo.15240501}{10.5281/zenodo.15240501}; raw datasets are excluded from the archive and reproduced by checksum-verified download scripts.  The third-party baselines are PGM-index and RadixSpline, fetched by the artifact scripts from their public source repositories.

\appendix

\section{Benchmark Diagnostics}
\label{app:benchmark-diagnostics}

The main text reports the exact finite-parameter table and one compact full-scale successful-lookup table.  The additional tables below document stress workloads, coarsened rank-spread diagnostics for the full-scale prototype, shadow routing overhead, and the repair-latency correlation used to interpret that diagnostic prototype.  They are not used as empirical proof of the accounting theorem.

\begin{table}[H]
\centering
\small
\caption{Stress-workload mean lookup latency in ns.  Each cell shows the fastest configuration in the family for that workload; repeated-run confidence intervals are shown when available.}
\label{tab:stress-workloads}
\resizebox{\linewidth}{!}{%
\begin{tabular}{llrrrr}
Dataset & Workload & Binary & Shadow & PGM & RS \\
\hline
Books & misses & 470.9{\scriptsize$\pm$24.8} & 751.5{\scriptsize$\pm$33.3} & 187.2{\scriptsize$\pm$7.9} & 194.3{\scriptsize$\pm$2.7} \\
Books & mixed & 498.1{\scriptsize$\pm$5.5} & 763.3{\scriptsize$\pm$30.3} & 185.6{\scriptsize$\pm$2.5} & 180.3{\scriptsize$\pm$10.7} \\
Books & zipf & 460.0{\scriptsize$\pm$38.5} & 672.2{\scriptsize$\pm$48.2} & 176.3{\scriptsize$\pm$11.7} & 180.4{\scriptsize$\pm$3.8} \\
Books & hotspot & 159.3{\scriptsize$\pm$7.5} & 138.3{\scriptsize$\pm$4.3} & 71.4{\scriptsize$\pm$5.7} & 60.2{\scriptsize$\pm$0.4} \\
Books & gaps & 61.9{\scriptsize$\pm$1.7} & 83.9{\scriptsize$\pm$2.2} & 56.4{\scriptsize$\pm$1.1} & 37.9{\scriptsize$\pm$0.6} \\
\hline
FB & misses & 462.7{\scriptsize$\pm$25.9} & 1292.6{\scriptsize$\pm$3.1} & 219.1{\scriptsize$\pm$12.0} & 420.7{\scriptsize$\pm$27.4} \\
FB & mixed & 461.0{\scriptsize$\pm$30.9} & 1270.6{\scriptsize$\pm$12.0} & 213.5{\scriptsize$\pm$8.3} & 413.8{\scriptsize$\pm$22.3} \\
FB & zipf & 450.2{\scriptsize$\pm$19.1} & 1302.5{\scriptsize$\pm$62.7} & 230.2{\scriptsize$\pm$3.4} & 401.3{\scriptsize$\pm$5.7} \\
FB & hotspot & 154.8{\scriptsize$\pm$26.7} & 201.7{\scriptsize$\pm$7.6} & 83.4{\scriptsize$\pm$3.1} & 127.2{\scriptsize$\pm$2.7} \\
FB & gaps & 163.0{\scriptsize$\pm$41.0} & 229.5{\scriptsize$\pm$14.4} & 74.2{\scriptsize$\pm$0.7} & 120.1{\scriptsize$\pm$4.8} \\
\hline
Wiki & misses & 484.8{\scriptsize$\pm$36.2} & 1319.4{\scriptsize$\pm$56.3} & 162.4{\scriptsize$\pm$12.3} & 166.1{\scriptsize$\pm$20.5} \\
Wiki & mixed & 472.4{\scriptsize$\pm$55.6} & 1291.5{\scriptsize$\pm$16.1} & 161.6{\scriptsize$\pm$5.7} & 146.1{\scriptsize$\pm$13.3} \\
Wiki & zipf & 475.7{\scriptsize$\pm$55.5} & 1265.5{\scriptsize$\pm$46.8} & 167.3{\scriptsize$\pm$18.7} & 176.7{\scriptsize$\pm$14.7} \\
Wiki & hotspot & 148.4{\scriptsize$\pm$7.9} & 208.9{\scriptsize$\pm$4.8} & 71.3{\scriptsize$\pm$7.3} & 60.7{\scriptsize$\pm$1.0} \\
Wiki & gaps & 52.1{\scriptsize$\pm$2.4} & 76.5{\scriptsize$\pm$1.6} & 51.1{\scriptsize$\pm$1.1} & 33.1{\scriptsize$\pm$0.8} \\
\hline
OSM & misses & 748.6{\scriptsize$\pm$34.4} & 2183.9{\scriptsize$\pm$100.0} & 305.5{\scriptsize$\pm$28.8} & 401.7{\scriptsize$\pm$26.7} \\
OSM & mixed & 769.8{\scriptsize$\pm$94.6} & 2175.9{\scriptsize$\pm$105.9} & 293.2{\scriptsize$\pm$3.4} & 384.1{\scriptsize$\pm$1.8} \\
OSM & zipf & 799.7{\scriptsize$\pm$186.8} & 2017.8{\scriptsize$\pm$8.3} & 301.6{\scriptsize$\pm$16.2} & 373.9{\scriptsize$\pm$6.6} \\
OSM & hotspot & 200.0{\scriptsize$\pm$26.6} & 284.4{\scriptsize$\pm$7.3} & 101.3{\scriptsize$\pm$1.9} & 111.1{\scriptsize$\pm$0.1} \\
OSM & gaps & 106.5{\scriptsize$\pm$5.0} & 308.7{\scriptsize$\pm$65.5} & 84.4{\scriptsize$\pm$2.1} & 44.5{\scriptsize$\pm$0.8} \\
\hline

\end{tabular}
}
\end{table}
\FloatBarrier

\begin{table}[H]
\centering
\small
\caption{Coarsened rank-spread diagnostics for the fastest shadow configuration on each dataset/workload.  Ratio is empirical residual entropy divided by the empirical log-window denominator; Support counts residual positions with non-negligible mass.  These are trace diagnostics, not exact theorem quantities.}
\label{tab:rank-spread-diagnostic}
\resizebox{\linewidth}{!}{%
\begin{tabular}{llrrrr}
Dataset & Workload & Index & Ratio & Support & Window \\
\hline
Books & hits & Shadow-O(1024) & 0.73 & 310.3 & 7637.5 \\
Books & misses & Shadow-O(1024) & 0.73 & 310.1 & 5970.9 \\
Books & mixed & Shadow-O(1024) & 0.73 & 310.7 & 5978.0 \\
Books & zipf & Shadow-O(1024) & 0.75 & 289.6 & 4614.2 \\
Books & hotspot & Shadow-R(1024) & 0.52 & 85.6 & 3214.8 \\
Books & gaps & Shadow-R(1024) & 0.10 & 2.1 & 98550.6 \\
\hline
FB & hits & Shadow-R(1024) & 0.64 & 4233.6 & 53052.0 \\
FB & misses & Shadow-R(1024) & 0.64 & 4233.5 & 53718.6 \\
FB & mixed & Shadow-O(1024) & 0.64 & 4259.5 & 52432.2 \\
FB & zipf & Shadow-R(1024) & 0.64 & 3386.7 & 47004.2 \\
FB & hotspot & Shadow-O(1024) & 0.44 & 699.9 & 71193.8 \\
FB & gaps & Shadow-R(256) & 0.34 & 111.2 & 58804.1 \\
\hline
Wiki & hits & Shadow-O(1024) & 0.61 & 1220.0 & 46887.2 \\
Wiki & misses & Shadow-O(1024) & 0.61 & 1219.9 & 46887.2 \\
Wiki & mixed & Shadow-O(1024) & 0.61 & 1228.9 & 47601.2 \\
Wiki & zipf & Shadow-O(1024) & 0.61 & 921.9 & 39556.4 \\
Wiki & hotspot & Shadow-O(1024) & 0.38 & 246.6 & 63819.6 \\
Wiki & gaps & Shadow-O(256) & 0.57 & 18.6 & 62161.6 \\
\hline
OSM & hits & Shadow-R(1024) & 0.23 & 41.9 & 5633128.2 \\
OSM & misses & Shadow-O(1024) & 0.46 & 3481.0 & 5634460.9 \\
OSM & mixed & Shadow-O(1024) & 0.46 & 3480.2 & 5632448.5 \\
OSM & zipf & Shadow-R(1024) & 0.24 & 45.3 & 3600901.4 \\
OSM & hotspot & Shadow-O(1024) & 0.39 & 452.0 & 1806917.8 \\
OSM & gaps & Shadow-R(1024) & 0.00 & 1.0 & 5707344.5 \\
\hline

\end{tabular}
}
\end{table}
\FloatBarrier

The rank-spread table is deliberately not used as a theorem proof.  On Books, FB, and Wiki, high hit/miss ratios say that the observed residual mass is spread enough for the radius surrogate to be informative on those traces.  The low adversarial-gap ratios, especially on OSM, show the opposite regime: the certified windows are large, but the residual answers concentrate, so the paper must use $\RGap_\M$ rather than a blanket $\log(1+\Delta)$ lower bound.

\begin{table}[H]
\centering
\small
\caption{Shadow routing overhead on successful lookups.  Ordered routing keeps the index small but pays ordered segment search; radix routing moves part of that cost into explicit directory bytes outside $B$.}
\label{tab:shadow-overhead}
\resizebox{\linewidth}{!}{%
\begin{tabular}{lrrrrrr}
Dataset & Ord MB & Rad MB & Ordered ns & Radix ns & Ord route & Rad route \\
\hline
Books & 0.07 & 2.16 & 743.5{\scriptsize$\pm$34.9} & 782.0{\scriptsize$\pm$38.2} & 10.0 & 2.0 \\
FB & 0.07 & 2.16 & 1351.8{\scriptsize$\pm$9.0} & 1336.2{\scriptsize$\pm$45.9} & 10.0 & 11.0 \\
Wiki & 0.07 & 2.16 & 1281.2{\scriptsize$\pm$9.8} & 1371.7{\scriptsize$\pm$43.7} & 10.0 & 2.0 \\
OSM & 0.07 & 2.16 & 2160.2{\scriptsize$\pm$60.9} & 2151.0{\scriptsize$\pm$77.9} & 10.0 & 2.2 \\

\end{tabular}
}
\end{table}
\FloatBarrier

Repair comparisons correlate with latency when the index family and routing overhead are comparable.  \Cref{tab:repair-latency-corr} reports Pearson correlations between average repair comparisons and mean latency over learned configurations.  Weak or unstable correlations identify cases where cache behavior, radix tables, directory lookup, or branch behavior dominate the last-mile comparison count.

\begin{table}[H]
\centering
\small
\caption{Pearson correlation between average repair comparisons and mean latency across benchmark configurations.}
\label{tab:repair-latency-corr}
\begin{tabular}{lrr}
Dataset & Hits learned & All workloads \\
\hline
Books & 0.86 & 0.69 \\
FB & 0.92 & 0.62 \\
Wiki & 0.96 & 0.80 \\
OSM & 0.98 & 0.76 \\

\end{tabular}
\end{table}
\FloatBarrier

\section{Proof of the Routing Lemmas}

\begin{proof}[Expanded proof of \cref{lem:routing-lb}]
Let $d_1,\ldots,d_m$ be the leaf depths, and let $q_j=2^{-d_j}$.  Since the tree is binary, Kraft's inequality yields $\sum_j q_j\le 1$.  Now
\[
\sum_j p_jd_j
=\sum_jp_j\log\frac{1}{q_j}
=\sum_jp_j\log\frac{1}{p_j}+\sum_jp_j\log\frac{p_j}{q_j}.
\]
If $Q=\sum_jq_j<1$, normalize to $\widehat q_j=q_j/Q$.  Then
\[
\sum_jp_j\log\frac{p_j}{q_j}
=\sum_jp_j\log\frac{p_j}{\widehat q_j}+\log\frac{1}{Q}
=\KL(p\|\widehat q)+\log\frac{1}{Q}\ge 0.
\]
If $Q=1$, the same argument holds without the normalization term.  Therefore $\sum_jp_jd_j\ge H_\mu(\Pi)$.
\end{proof}

\begin{proof}[Expanded proof of \cref{lem:routing-ub}]
Mehlhorn's construction builds an alphabetic search tree for ordered probabilities $p_1,\ldots,p_m$ with expected depth at most $H+2$, where $H=\sum_jp_j\log(1/p_j)$ \cite{mehlhorn1975obst}.  Our leaf intervals are ordered and nonoverlapping, so they satisfy the alphabetic constraint.  Applying the construction directly gives the claimed directory tree.
\end{proof}

\section{Proof of the Repair Lemmas}

\begin{proof}[Expanded proof of \cref{lem:repair-lb}]
Fix a leaf with promised error radius $\Delta$ and predicted rank $\widehat r$.  The true predecessor rank lies in
\[
  [\widehat r-\Delta,\widehat r+\Delta]\cap\{0,1,\ldots,n\}.
\]
Ignoring boundary effects, there are $M=\Theta(1+\Delta)$ possible predecessor answers across the gaps in that window.  For each possible answer, choose a query value falling into the corresponding gap while preserving the same prediction guarantee.  A comparison decision tree that returns the exact predecessor must distinguish these $M$ cases.  A binary decision tree with depth $t$ has at most $2^t$ leaves, so $t\ge\log M=\Omega(\log(1+\Delta))$ in the worst case.
\end{proof}

\begin{proof}[Expanded proof of \cref{lem:repair-ub}]
The candidate predecessor lies among $O(1+\Delta)$ consecutive positions around the prediction.  Binary search over the clipped candidate interval determines the exact predecessor or rank using $O(\log(1+\Delta))$ comparisons.  If the data structure stores only the predicted rank and not the exact interval endpoints, exponential search around the predicted position finds a containing interval of size $O(1+\Delta)$ and then binary search completes the correction within the same asymptotic bound.
\end{proof}

\section{Expanded Proof of the Static Theorem}

We spell out both directions, including the finite repair-program construction used only in the residual upper bound.  For any routed piecewise index $D$, write
\[
  T_D(q)=T_D^{\mathrm{route}}(q)+T_D^{\mathrm{repair}}(q)+O(1),
\]
where the additive constant covers local model evaluation.  If $q\in I_j$, the routing path has length at least the depth $d_j$ of that leaf, and the repair phase must identify the residual exact answer left by the local predictor.  Therefore
\[
\begin{aligned}
  \E[T_D]
  &\ge \sum_j p_jd_j+\sum_jp_j\Rep_\mu(I_j,h_j,\Delta_j)-O(1)\\
  &\ge H_\mu(\Pi)+\sum_jp_j\Rep_\mu(I_j,h_j,\Delta_j)-O(1),
\end{aligned}
\]
where the last inequality is \cref{lem:routing-lb} and the conditional repair term follows from \cref{lem:residual-entropy}.  The tuple implemented by $D$ is feasible in \eqref{eq:rgap-def}, hence
\[
  \E[T_D]\ge \RGap_\M(S,\mu,B)-O(1)
  =\Omega(\RGap_\M(S,\mu,B)).
\]
This proves the residual lower-bound direction without a rank-spread assumption.  If $D$ is additionally rank-spread, \cref{lem:rank-spread-repair} substitutes
$\Rep_\mu(I_j,h_j,\Delta_j)\ge c\log(1+\Delta_j)$, and feasibility for
\eqref{eq:gap-def} gives the radius lower bound.

For the residual upper bound, fix $\eta>0$ and select a feasible tuple
$F=(\Pi,\{h_j,\Delta_j\}_j)$ such that
\[
  H_\mu(\Pi)+\sum_jp_j\Rep_\mu(I_j,h_j,\Delta_j)
  \le \RGap_\M(S,\mu,B)+\eta/3.
\]
An alphabetic directory for $\Pi$ has expected depth at most
$H_\mu(\Pi)+2$.  Fix a leaf and a positive-probability finite transcript
$y$.  Let $P_y$ be the conditional distribution on its positive-mass
answers and let $K_y$ include every answer attainable in the required query
domain under that transcript.  Insert the zero-mass attainable answers in
their order with total auxiliary mass $\varepsilon_y$, and multiply $P_y$ by
$1-\varepsilon_y$.  Applying the alphabetic-code upper bound to this full
ordered distribution gives a comparison tree that is correct on all $K_y$
answers.  When it is evaluated under $P_y$, its expected comparisons are at
most
\[
  H(P_y)+2+\log\frac{1}{1-\varepsilon_y}.
\]
Choose the finitely many $\varepsilon_y$ so that the expectation of the last
term over all positive-mass transcripts is at most $\eta/3$.  For attainable
transcripts of zero workload mass, choose any correct finite tree; this adds
no expected comparisons.

Dispatching to these transcript-indexed trees yields
\[
\begin{aligned}
  \E[T_F]
  &\le H_\mu(\Pi)+2+
       \sum_jp_j\Rep_\mu(I_j,h_j,\Delta_j)+2+\eta/3+O(1)\\
  &\le \RGap_\M(S,\mu,B)+O(1)+\eta .
\end{aligned}
\]
The comparison bound is independent of the table size.  Materializing the
trees, however, requires $\Theta(\sum_yK_y)$ nodes in general; those nodes
are outside $B$ under the accounting convention and are exactly the
non-atom resource made visible in the statement and experiments.

Finally, the radius upper bound in \cref{cor:rank-spread-gap} does not require
that transcript-indexed coding construction.  Take a near-minimizer of
\eqref{eq:gap-def}, build its certified predictors and alphabetic directory,
and binary-search each certified window.  By \cref{lem:repair-ub}, its
expected cost is
\[
  O\!\left(H_\mu(\Pi)+\sum_jp_j\log(1+\Delta_j)+1\right)
  =O(\Gap_\M(S,\mu,B)+1).
\]

\section{Power-Law Optimization Details}

We prove \cref{thm:power-law} with the integer floor and the full-repair boundary included.  Put
\[
  y_j=1+\Delta_j,
  \qquad 1\le y_j\le 1+R_j,
  \qquad \Beff=B+m.
\]
For non-boundary predictors, the normalized profile condition gives
\[
  1+\Comp_{I_j,\M}(\Delta_j)\asymp 1+\kappa_j y_j^{-\alpha}.
\]
Thus a budget $B$ implies an effective envelope constraint
\[
  \sum_{j=1}^m \kappa_j y_j^{-\alpha}=O(\Beff)
\]
for all leaves not using the empty predictor.  Conversely, any integer radii satisfying the envelope with a sufficiently small hidden constant yield a feasible allocation after changing $B$ by a constant factor.  This is the only place where the additive $m$ in $\Beff$ enters.

\begin{proof}[Detailed proof of \cref{thm:power-law}]
For the upper bound, define
\[
  y_j^*=\min\left\{1+R_j,\max\left\{1,\left(\frac{c\kappa_j}{\Beff p_j}\right)^{1/\alpha}\right\}\right\}
\]
with a constant $c>0$ chosen from the profile constants.  If $y_j^*<1+R_j$, then
\[
  \kappa_j (y_j^*)^{-\alpha}\le c^{-1}\Beff p_j.
\]
Summing over such leaves gives at most $c^{-1}\Beff$.  If $y_j^*=1+R_j$, we use the empty predictor on that leaf and pay no local atoms; its contribution is handled by full repair.  Rounding $y_j^*-1$ to the nearest integer in $[0,R_j]$ changes $y_j$ by at most a factor of two except in the constant-radius regime, where the repair cost changes by $O(1)$.  Hence the rounded allocation is feasible up to a constant rescaling of $B$ and has repair cost
\[
  O\left(\sum_j p_j\log\left(1+
\min\left\{R_j,\max\left\{1,\left(\frac{\kappa_j}{\Beff p_j}\right)^{1/\alpha}\right\}\right\}\right)\right).
\]

For the lower bound, relax the integer radii to real variables $y_j\in[1,1+R_j]$ and replace each local profile by the lower envelope.  This can only decrease the optimum.  The relaxed repair problem is
\[
  \min \sum_j p_j\log y_j
  \quad\text{subject to}\quad
  \sum_{j\in A}\kappa_j y_j^{-\alpha}\le C\Beff,
  \qquad 1\le y_j\le 1+R_j,
\]
where $A$ denotes leaves not assigned full repair and $C$ is an absolute constant.  The Karush--Kuhn--Tucker conditions give, for every non-boundary coordinate,
\[
  \frac{p_j}{y_j}=\lambda\alpha\kappa_j y_j^{-\alpha-1},
  \qquad\text{hence}\qquad
  y_j^\alpha=\frac{\lambda\alpha\kappa_j}{p_j}.
\]
Coordinates for which this value is below one are fixed at the lower boundary, and coordinates for which it exceeds $1+R_j$ are fixed at full repair.  Choosing $\lambda$ so that the relaxed envelope is tight gives the same scale as the upper allocation, with $\Beff$ replacing $B$ because of the atom floor.  Substitution yields the capped expression in \cref{thm:power-law}.  In the interior regime the caps are inactive and
\[
  \log y_j=\Theta\left(\frac{1}{\alpha}\logplus\frac{\kappa_j}{\Beff p_j}\right),
\]
which gives the positive-log formula.  The discrete optimum differs from the relaxed optimum by at most constant factors because rounding changes every non-boundary $y_j$ by a constant factor and the boundary cases contribute only $O(p_j)$ per leaf.  Adding the fixed routing term $H_\mu(\Pi)$ completes the proof.
\end{proof}

\section{Dynamic Proof Details}

\begin{proof}[Expanded proof of \cref{thm:dynamic-upper}]
The space bound is immediate because each live key is stored in exactly one level and each active level uses a linear-space static structure.  For updates, rebuilding a parent from $\beta$ full children of size $\Theta(\beta^i)$ costs $\Theta(\beta^{i+1})$.  Such a rebuild is triggered after $\Theta(\beta^i)$ lower-level updates, so the amortized contribution of level $i$ is $\Theta(\beta)$.  There are $\Theta(\log_\beta n)$ levels, yielding $\Theta(\beta\log_\beta n)$ amortized update time.

For queries, the structure queries every active level and combines predecessor candidates by taking their maximum below $q$, or combines rank contributions by summing per-level ranks.  The expected cost at level $i$ is bounded by the static residual-entropy cost for $S_i$ under the induced distribution, plus a constant.  Summing over active levels proves the theorem.
\end{proof}

\begin{proof}[Expanded proof of \cref{thm:level-lb}]
Let the active levels be $L_0,\ldots,L_h$.  Choose one live key $y_i$ inside each level $L_i$ so that $y_0<\cdots<y_h<q$.  Suppose an exact algorithm does not inspect $L_t$.  The algorithm's observed transcript is unchanged if $L_t$ is modified to contain a key $z$ with $y_h<z<q$, while all inspected levels remain fixed.  The correct predecessor changes from $y_h$ to $z$, but the algorithm cannot distinguish the two instances.  This contradicts exactness.  Hence all active levels must be inspected.
\end{proof}

\begin{proof}[Expanded proof of \cref{prop:tombstones}]
After $N/2$ insertions and $N/2$ deletions inside the same query range, a tombstone-based structure that has not cleaned the range contains $\Theta(N)$ stale markers there.  Place $k$ live keys in the same range.  An exact range query must identify which encountered markers correspond to deleted keys and which entries are live.  Otherwise it cannot certify the exact output.  Therefore the range query must examine $\Omega(N+k)$ items in the worst case, unless cleanup work has already been paid elsewhere.
\end{proof}

\begin{proof}[Expanded proof of \cref{thm:rank-barrier}]
Represent the dynamic subset-rank problem as a dynamic ordered set on universe $[n]$.  A bit $1$ at position $i$ means that key $i$ is present; a bit $0$ means absent.  Insertions and deletions are bit flips.  A rank query at $i$ returns the number of present keys in $[1,i]$, exactly the prefix sum of the bitvector.  Any exact dynamic learned index supporting rank solves this problem as a special case.  Since the cell-probe model gives computation for free, any lower bound for dynamic subset rank or dynamic prefix sums applies even if prediction is arbitrarily powerful.  The standard $w=\Theta(\log n)$ consequence is the $\Omega(\log n/\log\log n)$ barrier for one of update or query time.
\end{proof}

\end{document}